\documentclass[12pt]{article}

\usepackage{amsmath, amssymb, amsthm}  % Math packages
\usepackage{fancyhdr}                  % For custom headers and footers
\usepackage{graphicx}                  % For including graphics
\usepackage{hyperref}                  % For hyperlinks in the document
\usepackage{geometry}                  % To change the page dimensions
\usepackage{natbib}
\usepackage{subcaption,verbatim}
\usepackage{xcolor}
\usepackage{color}
\usepackage{blkarray}
\usepackage[utf8]{inputenc}
\usepackage[english]{babel}
\usepackage[normalem]{ulem}
\usepackage{xr-hyper}

\newcommand{\bM}{\mathbf{M}}

\newcommand{\ke}{\scriptscriptstyle(k+1)}
\newcommand{\ka}{\scriptscriptstyle(k)}
\newcommand{\LL}{\mathbf{L}}

\def\bSig\boldsymbol{\Sigma}

\newtheorem{theorem}{Theorem}

\newcommand{\sumas}{\sum^n_{i=1}}

\newcommand{\ii}{i\in\{1,\ldots,n\}}

\newcommand{\by}{{\bf Y}_j}
\newtheorem{Lemma}{Lemma}
\newtheorem{proposition}{Proposition}

\newtheorem{corollary}{Corollary}
\newtheorem{remark}{Remark}

\newcommand{\bmu}{\boldsymbol{\mu}}

\newcommand{\bSigma}{\boldsymbol{\Sigma}}

\newcommand{\bbeta}{\mbox{\boldmath $\beta$}}
\newcommand{\btheta}{\mbox{\boldmath $\theta$}}
\newcommand{\bTheta}{\mbox{\boldmath $\Theta$}}

\newcommand{\bDelta}{\mbox{\boldmath $\Delta$}}

\newcommand{\x}{\boldsymbol{x}}
\newcommand{\w}{\boldsymbol{w}}

\newcommand{\B}{\boldsymbol{B}}

\newcommand{\bPsi}{\mbox{\boldmath $\Psi$}}

\newcommand{\bOmega}{\mbox{\boldmath $\Omega$}}

\newcommand{\bgamma}{\mbox{\boldmath $\gamma$}}

\newcommand{\bC}{\boldsymbol{C}}

\newcommand{\Y}{\mathbf{y}}
\newcommand{\bY}{\mathbf{Y}}

\newcommand{\bD}{\boldsymbol{D}}

\newcommand{\X}{\boldsymbol{X}}

\newcommand{\bV}{\boldsymbol{V}}

\newcommand{\tr}{\textrm{tr}}

\newcommand{\bm}{\mathbf{m}}

\usepackage{xcolor}

\title{Multiple Heckman Selection Model}

\author{{ Heeju Lim$^{1}$,  Carlos A.R. Diniz$^{2}$, Ofer Harel$^{1}$ and Victor H. Lachos$^{1}$} \\
{\small$^{(1)}$ Department of Statistics, University of Connecticut, CT-06269, U.S.A.}\\
{\small$^{(2)}$ Department of Statistics, Federal University of São Carlos, SP, Brazil.}
}

\begin{document}
\maketitle
\begin{abstract}
We introduce a novel matrix-variate extension of the Heckman selection model to accommodate multiple outcomes, providing a flexible and natural generalization of classical selection models for matrix-valued data. By relying on the matrix normal distribution, the proposed model captures dependencies across both rows and columns while accounting for selection bias. An Expectation/Conditional Maximization (ECM) algorithm is developed, yielding closed-form updates for all model parameters. We investigate key theoretical properties, including the connection between sample selection models and the recently developed multivariate unified skew-normal (SUN) distribution. The performance of the proposed approach is assessed through simulation studies, and its practical utility is illustrated using two real datasets. The proposed method is implemented in the \textsf{R} package \texttt{mvHeckman}.
\end{abstract}


%\bibliographystyle{unsrtnat}
%\bibliography{references}  %%% Uncomment this line and comment out the ``thebibliography'' section below to use the external .bib file (using bibtex) .

\section{Introduction}
\label{Intro}

{
Sample selection bias arises when outcomes are observed only for a subset of the population and the selection mechanism depends on unobserved factors, leading to biased inference. To address this issue, \citet{heckman1974shadow} introduced a parametric framework based on a joint normality assumption. However, this assumption is often violated in practice, particularly when the data exhibit heavy tails or skewness. To accommodate such features, \citet{marchenko2012heckman} proposed the selection-$t$ model, which replaces the normal distribution with a heavier-tailed alternative. More recently, \citet{lim2025heckman} introduced a selection model based on the contaminated normal distribution, providing a flexible framework that captures both typical observations and deviations such as outliers. These developments reflect ongoing efforts to construct more robust selection models that better represent the distributional characteristics of empirical data.
\\
\indent While the classical Heckman model provides a useful framework for addressing selection bias, it is primarily designed for a single outcome variable. In many empirical applications, however, multiple outcomes are observed simultaneously and may be interrelated. For example, economic behaviors, health indicators, or consumption patterns often exhibit complex dependencies across variables. Modeling such outcomes separately may ignore these dependencies, potentially leading to inefficient inference or biased conclusions. These considerations motivate the extension of sample selection models to multivariate settings.
\\
\indent Several studies have extended the classical selection framework to accommodate multiple outcomes. For example, \citet{yen2005multivariate} developed a multivariate sample selection model to jointly analyze related demand systems. Subsequent work has examined the statistical properties and estimation of such models, highlighting the challenges associated with high-dimensional dependence structures and the complexity of the likelihood function \citep{tauchmann2010consistency}. More recent contributions have explored extensions of the selection mechanism itself, allowing for multiple selection rules and more flexible participation processes beyond the standard binary framework \citep{rezaee2022sample}. Estimation in these models typically relies on full maximum likelihood, simulation-based methods, or extensions of the two-step procedure, which can be computationally demanding in high-dimensional settings.
\\
\indent Despite these advances, most existing approaches treat multivariate outcomes as vector-valued and capture dependence primarily through unstructured covariance matrices, without explicitly exploiting potential structure across dimensions. To address these limitations, we propose a matrix-variate sample selection model that explicitly accounts for such structured dependencies while allowing for deviations from normality. The proposed framework extends existing selection models by incorporating matrix-valued responses and a flexible error specification. We further investigate key theoretical properties of the model, showing that the induced distribution of the observed outcomes is closely related to the multivariate unified skew-normal (SUN) family. This connection provides additional insight into the distributional effects of sample selection and supports the theoretical foundation of the proposed approach. In addition, we develop an efficient ECM algorithm for likelihood-based inference and provide an implementation of the proposed method in the \texttt{mvHeckman} R package. The performance of the proposed approach is evaluated through simulation studies, where we examine parameter recovery and compare its performance with the classical univariate Heckman model.
\\
\indent The remainder of the paper is organized as follows. Section~2 reviews the classical Heckman selection model and its key properties. Section~3 introduces the multivariate extension and presents its matrix formulation. Section~4 develops the proposed matrix-variate selection model, including the likelihood function and estimation procedure based on the ECM algorithm. Section~5 presents simulation studies to evaluate the performance of the proposed method. Section~6 illustrates the proposed methodology through applications to labor supply data and biomedical measurements, highlighting its broad applicability across domains. Taken together, these developments provide a coherent framework for modeling multivariate outcomes under structured dependence and sample selection. Concluding remarks and potential directions for future research are discussed in Section~7. Technical derivations and supporting material are gathered in the Appendix.
}

\section{Preliminaries}\label{sec:prel}

In this section, we present key notation and preliminary results used throughout the manuscript. We begin by introducing the notation used throughout the paper. A random matrix of dimension $p \times q$ is denoted by $\boldsymbol{\mathcal{X}}$, and its realization by $\mathcal{X}$. A random vector of length $pq$ is denoted by $\mathbf{X}$, and its realization by $\mathbf{x}$. The notation $|\cdot|$ denotes the determinant of a square matrix, $\operatorname{tr}(\cdot)$ denotes its trace, and $\operatorname{vec}(\cdot)$ denotes the vectorization of a matrix (i.e., stacking its columns into a single vector). The symbol $\otimes$ denotes the Kronecker product, and $U \perp Z$ denotes independence between $U$ and $Z$. 

Throughout this paper, $\bY \sim \mathcal{N}_p(\bmu, \bSigma)$ denotes a $p$-variate random vector following a multivariate normal distribution with mean vector $\bmu$ and covariance matrix $\bSigma$. Its probability density function (pdf) and cumulative distribution function (cdf) are denoted by $\phi_p(\cdot \mid \bmu, \bSigma)$ and $\Phi_p(\cdot \mid \bmu, \bSigma)$, respectively. We also write $\Phi_p(\mathbf{a}, \mathbf{b} \mid \bmu, \bSigma)$ for $\Pr(\mathbf{a} < \bY < \mathbf{b})$. When $p = 1$, we omit the subscript $p$; in this case, if $\mu = 0$ and $\sigma^2 = 1$, we write $\phi(\cdot)$ and $\Phi(\cdot)$ for the pdf and cdf, respectively.

\subsection{Matrix-variate normal distribution}
\label{sec:MT}

The matrix-variate normal (MVN) distribution is a fundamental distribution in the matrix-variate setting, mainly due to its mathematical tractability \citep{gupta1999matrix}. A $p \times q$ random matrix $\boldsymbol{\mathcal{Y}}$ is said to follow an MVN distribution with mean matrix $\bM$ and covariance matrices $\bSigma$ and $\bPsi$ of dimensions $p \times p$ and $q \times q$, respectively, if its pdf is given by
\begin{equation}
f(\mathcal{Y} \mid \bM, \bSigma, \bPsi)
= \frac{1}{(2\pi)^{pq/2} |\bPsi|^{p/2} |\bSigma|^{q/2}}
\exp\left\{
-\frac{1}{2} \operatorname{tr}\left[\bPsi^{-1}(\mathcal{Y}-\bM)^{\top}\bSigma^{-1}(\mathcal{Y}-\bM)\right]
\right\}.
\label{eq:MVN}
\end{equation}

We write $\boldsymbol{\mathcal{Y}} \sim \mathcal{N}_{p \times q}(\bM, \bSigma, \bPsi)$.

%The following lemma will be useful in subsequent developments.

\begin{comment}
\begin{Lemma}\label{identity_1} 
Let $\textbf{P}\in\mathbb{R}^{m\times n}$, $\textbf{Q}\in\mathbb{R}^{n\times p}$, $\textbf{R}\in\mathbb{R}^{p\times q}$ and $\textbf{S}\in\mathbb{R}^{q\times m}$. Then, the following identity holds: 
\begin{align*} 
\operatorname{tr}(\textbf{P}\textbf{Q}\textbf{R}\textbf{S}) & = \operatorname{vec}(\textbf{S}^{\top})^{\top}(\textbf{R}^{\top}\otimes\textbf{P})\operatorname{vec}(\textbf{Q}). 
\end{align*} 
\end{Lemma}
\begin{proof} First, we have from the properties of the Kronecker product and $\operatorname{vec}$ operator that for any conformable matrices  $\operatorname{tr}(\textbf{A}^{\top}\textbf{B})=\operatorname{vec}(\textbf{A})^{\top}\operatorname{vec}(\textbf{B})$ and $\operatorname{vec}(\textbf{P}\textbf{Q}\textbf{R})=(\textbf{R}^{\top}\otimes\textbf{P})\operatorname{vec}(\textbf{Q})$ - see \citep{petersen2008matrix}. Thus,
\begin{align*} \operatorname{tr}(\textbf{P}\textbf{Q}\textbf{R}\textbf{S}) & = \operatorname{tr}(\textbf{S}\textbf{P}\textbf{Q}\textbf{R})=\operatorname{tr}((\textbf{S}^{\top})^{\top}\textbf{P}\textbf{Q}\textbf{R})\\ & = \operatorname{vec}(\textbf{S}^{\top})^{\top}\operatorname{vec}(\textbf{P}\textbf{Q}\textbf{R})\\ & = \operatorname{vec}(\textbf{S}^{\top})^{\top}(\textbf{R}^{\top}\otimes\textbf{P})\operatorname{vec}(\textbf{Q}), 
\end{align*} 
which concluded the proof.
\end{proof}
\end{comment}

An equivalent representation of the MVN distribution can be obtained via vectorization. Specifically,
\begin{equation}
\boldsymbol{\mathcal{Y}} \sim \mathcal{N}_{p \times q}(\bM, \bSigma, \bPsi)
\Longleftrightarrow
\operatorname{vec}(\boldsymbol{\mathcal{Y}}) \sim \mathcal{N}_{pq}(\operatorname{vec}(\bM), \bPsi \otimes \bSigma).
\label{eq:iff}
\end{equation}
%This result follows from standard properties of the $\operatorname{vec}$ operator and the Kronecker product (see, e.g., \citealp{petersen2008matrix}). 

It is also important to note an identifiability issue involving the two covariance matrices. Indeed, $\bPsi \otimes \bSigma = \bPsi^* \otimes \bSigma^*$ if $\bSigma^* = a \bSigma$ and $\bPsi^* = a^{-1} \bPsi$, for any $a > 0$. Therefore, $\bSigma$ and $\bPsi$ are identifiable only up to a multiplicative constant. Several approaches have been proposed in the matrix-variate literature to address this issue; see, e.g., \citep{lachos2025}.  In this paper, we impose the constraint $\Sigma_{pp} = 1$, a standard normalization in Heckman-type selection models, which ensures identifiability. The regular Heckman-type selection framework is introduced in the next section.

\subsection{The Heckman selection model}\label{SLnModel}

In many applied contexts, missing or incomplete data arise due to non-random sample selection mechanisms. The Heckman selection model consists of a linear
equation for the outcome, and a Probit equation for the sample selection mechanism. The outcome equation is
\begin{equation}\label{1HS}
Y_{1i}=\x^{\top}_i\bbeta+\epsilon_{1i},
\end{equation}
and the sample selection mechanism is characterized by the following latent linear equation, for $i \in \{1, \ldots, n\}$,
\begin{equation}\label{2HS}
Y_{2i}=\w^{\top}_i\bgamma+\epsilon_{2i}.
\end{equation}
The vectors $\bbeta\in {\mathbb R}^{p_1}$ and 
$\bgamma \in {\mathbb R}^{q_1}$ are unknown regression parameters. $\x^{\top}_i=(x_{i1},\ldots,x_{ip_1})$ and  $\w^{\top}_i=(w_{i1},\ldots,w_{iq_1})$ are known characteristics. The covariates in $\x_i$ and $\w_i$ may overlap with each other, and the exclusion restriction holds when at least
one of the elements of $\w_i$ are not in $\x_i$.  The indicator for sample selection is $C_i = \mathbb{I}(Y_{2i} > 0)$. Let $V_{1i}$ be the observed outcome; we observe $V_{1i}$ if and only if $C_i = 1$, i.e., $Y_{1i} = V_{1i}$ if $C_i = 1$, and $Y_{1i}$ is unobserved (missing) otherwise.

Under the assumption that the error terms follow a bivariate normal distribution, the Heckman selection model \citep{heckman1979sample} yields the selection normal model (SLn), in which
\begin{equation}\label{nerror}
\begin{pmatrix}
\epsilon_{1i}  \\
\epsilon_{2i}
\end{pmatrix}
\sim \mathcal{N}_{2}\left( \mathbf{0},
\bSigma
\right), \quad 
\bSigma=
\begin{pmatrix}
\sigma^2 & \rho\sigma \\
\rho\sigma & 1
\end{pmatrix}.
\end{equation}
The second diagonal element is fixed to 1 to ensure identifiability. The SLn model defined by \eqref{1HS}--\eqref{nerror} is known as the Type~2 Tobit model in the econometrics literature and is often referred to as the Heckman model. When $\rho = 0$, there is no selection effect, implying that the outcomes are missing at random and that the observed sample is representative of the population given the covariates.

Under the bivariate normal assumption, the mean equation for the outcomes of the selected samples is
\begin{equation}\label{correH}
\mathbb{E}\left[Y_{1i}\mid C_i=1,\x_i,\w_i\right]=\x^{\top}_i\bbeta+\rho\sigma \lambda(\w^{\top}_i{\bgamma}),
\end{equation}
where $\lambda(a)={\phi(a)}/{\Phi(a)}$ is the inverse Mills ratio. Therefore, ignoring the selection mechanism leads to model misspecification, since the mean equation for the observed outcomes is given by $\x_i^{\top}\bbeta$ augmented with the nonlinear selection correction term $\rho\sigma \lambda(\w_i^{\top}\bgamma)$. The ML estimates of the SLn model can be calculated by Newton-Raphson iteration or the EM algorithm as discussed by \citet{zhao2020new} and \cite{lachos2021heckman}.

Ignoring censoring for the moment, suppose that we have observations on $n$ independent individuals
\begin{equation}
%\Y_1, \ldots , \Y_n \ind \mathcal{N}_2(\bmu_i,\bSigma), \label{modeleq}
\Y_i \sim \mathcal{N}_2(\bmu_i,\bSigma), \; i \in \{ 1, \ldots , n\}, \label{modeleq}
\end{equation}
where $\Y_i=(Y_{1i},Y_{2i})^{\top}$ is the bivariate response vector for sample unit $i$, $$\bmu_i=\X_{ic}\bbeta_c,\,\,\,\X_{ic}=\left(\begin{array}{cc}
\x^{\top}_i & 0 \\
0 & \w^{\top}_i \\
\end{array}\right)\in \mathbb{R}^{2 \times (p_1 + q_1)},\,\,\,\bbeta_c=\left(\begin{array}{c}
\bbeta \\
\bgamma\\
\end{array}\right)\in \mathbb{R}^{(p_1 + q_1)}$$
and the dispersion matrix $\bSigma$ depends on an 
unknown parameter vector $(\sigma,\rho)$. 
\cite{lachos2021heckman} considers the approach proposed by \citet{vaida2009fast} and \citet{Matos11} to represent the model within the framework of a censored linear model. 

Thus, for the $i$th subject, the observed data consist of $(V_{1i}, C_i)$, where $C_i = \mathbb{I}(Y_{2i} > 0)$ and $V_{1i}$ is observed if $C_i = 1$ and unobserved otherwise.

\subsection{The likelihood function}\label{Likelihood_tMLC}

To obtain the likelihood function of the SLn model, first note that if $C_i=1$, then $Y_{1i}\sim \mathcal{N}(\x^{\top}_i\bbeta,\sigma^2)$ and $Y_{2i}\mid Y_{1i}=V_{1i}\sim \mathcal{N}(\mu_{c},\sigma^2_c)$, where
\[
\mu_c=\w^{\top}_i\bgamma+\frac{\rho}{\sigma}(V_{1i}-\x^{\top}_i\bbeta),
\quad
\sigma^2_c=1-\rho^2.
\]

Thus, the contribution to the likelihood for $C_i=1$ is
\[
\phi\left(V_{1i}\mid \x^{\top}_i\bbeta,\sigma^2\right)
\Phi\left(\frac{\mu_{c}}{\sigma_c}\right).
\]

If $C_i=0$, then the contribution to the likelihood is
\[
\Phi(-\w^{\top}_i\bgamma).
\]

Therefore, the likelihood function of $\btheta=(\bbeta^{\top},\bgamma^{\top},\sigma^2,\rho)^{\top}$ is
\begin{align}
L(\btheta\mid\bV,\bC)
=
\prod_{i=1}^{n}
\left\{
\phi(V_{1i}\mid\x^{\top}_i\bbeta,\sigma^2)
\Phi\left(\frac{\mu_{c}}{\sigma_c}\right)
\right\}^{C_i}
\left\{
\Phi(-\w^{\top}_i\bgamma)
\right\}^{1-C_i}.
\label{equ8.1}
\end{align}

The log-likelihood is given by $\ell(\btheta)=\log L(\btheta\mid\bV,\bC)$. Maximum likelihood estimation can be carried out via Newton--Raphson or EM-type algorithms, as discussed in \citet{zhao2020new} and \citet{lachos2021heckman}.

\section{The multiple Heckman selection model}\label{sec:SLnMultivariate}

In many applied settings, the variable of interest is not a single outcome but a set of related outcomes. This situation arises, for instance, in multivariate longitudinal studies or multichannel biomedical experiments, where multiple outcomes are observed for each subject. Often, each outcome is subject to a different selection process, leading to data that is missing not at random (MNAR). To address this issue, we extend the classical Heckman model to accommodate multiple correlated outcomes subject to distinct selection processes.

Let \( Y_{r1i} \) denote the \( r \)th outcome for individual $i \in \{ 1, \ldots , n\}$, and \( Y_{r2i} \) the associated latent selection variable, for $r = 1, \ldots, R$. The outcome $Y_{r1i}$ is observed if and only if $C_{ri} = \mathbb{I}(Y_{r2i} > 0)$.

The model for each outcome is given by
\begin{align}
Y_{r1i} &= \mathbf{x}_{ri}^\top \boldsymbol{\beta}_r + \varepsilon_{r1i}, \label{eq:multi_HS1} \\
Y_{r2i} &= \mathbf{w}_{ri}^\top \boldsymbol{\gamma}_r + \varepsilon_{r2i}, \label{eq:multi_HS2}
\end{align}
where \( \mathbf{x}_{ri} \in \mathbb{R}^{p_r} \) and \( \mathbf{w}_{ri} \in \mathbb{R}^{q_r} \) are covariates associated with the $r$th outcome and selection mechanism, respectively. The dimension of the vectors \( \boldsymbol{\beta}_r \in \mathbb{R}^{p_r} \) and \( \boldsymbol{\gamma}_r \in \mathbb{R}^{q_r} \) can vary across outcomes. The error terms \( \varepsilon_{r1i} \) and \( \varepsilon_{r2i} \) are jointly distributed and capture the dependence between the outcome and selection equations.

\subsection{Matrix formulation} 
\label{sec:matrix_form}

The model can be expressed in matrix form by organizing the outcomes and selection variables into a matrix \( \boldsymbol{\mathcal{Y}}_i \in \mathbb{R}^{2 \times R} \), where the first row contains the outcomes \( (Y_{11i}, \dots, Y_{R1i}) \), and the second row contains the corresponding latent selection variables \( (Y_{12i}, \dots, Y_{R2i}) \). We observe $n$ such matrices:
\[
\boldsymbol{\mathcal{Y}}_i =
\begin{pmatrix}
Y_{11i} & Y_{21i} & \cdots & Y_{R1i} \\
Y_{12i} & Y_{22i} & \cdots & Y_{R2i}
\end{pmatrix}, \quad i = 1, \ldots, n.
\]

Let \( \mathbf{x}_{ri} \in \mathbb{R}^{p_r} \) and \( \mathbf{w}_{ri} \in \mathbb{R}^{q_r} \), for $r = 1, \ldots, R$, denote the covariates associated with the outcome and selection equations, respectively. Define the block-structured covariate matrix
\[
\mathbf{Z}_i =
\begin{pmatrix}
\mathbf{x}_{1i}^\top & \mathbf{0}_{1 \times q_1} & \mathbf{x}_{2i}^\top & \mathbf{0}_{1 \times q_2} & \cdots & \mathbf{x}_{Ri}^\top & \mathbf{0}_{1 \times q_R} \\
\mathbf{0}_{1 \times p_1} & \mathbf{w}_{1i}^\top & \mathbf{0}_{1 \times p_2} & \mathbf{w}_{2i}^\top & \cdots & \mathbf{0}_{1 \times p_R} & \mathbf{w}_{Ri}^\top
\end{pmatrix}
\in \mathbb{R}^{2 \times (p + q)},
\]
where $p = \sum_{r=1}^{R} p_r$ and $q = \sum_{r=1}^{R} q_r$.

Let
\begin{equation}\label{betac}
\boldsymbol{\beta}^c_r =
\begin{pmatrix}
\boldsymbol{\beta}_r \\
\boldsymbol{\gamma}_r
\end{pmatrix}
\in \mathbb{R}^{p_r + q_r}, \quad r = 1, \ldots, R,
\end{equation}
denote the vector of regression parameters associated with the $r$th outcome. The combined regression coefficients can be written in block-diagonal form as
\[
\boldsymbol{B} =
\begin{pmatrix}
\boldsymbol{\beta}^c_1 & \mathbf{0} & \cdots & \mathbf{0} \\
\mathbf{0} & \boldsymbol{\beta}^c_2 & \cdots & \mathbf{0} \\
\vdots & \vdots & \ddots & \vdots \\
\mathbf{0} & \mathbf{0} & \cdots & \boldsymbol{\beta}^c_R
\end{pmatrix}
\in \mathbb{R}^{(p+q) \times R}.
\]

The model can then be written as
\[
\boldsymbol{\mathcal{Y}}_i = \mathbf{Z}_i \boldsymbol{B} + \boldsymbol{\mathcal{E}}_i,
\]
where \( \boldsymbol{\mathcal{E}}_i \in \mathbb{R}^{2 \times R} \) is an error matrix assumed to follow a matrix-variate normal distribution:
\begin{equation}\label{errorM}
\boldsymbol{\mathcal{E}}_i \sim \mathcal{N}_{2 \times R}(\mathbf{0}, \boldsymbol{\Sigma}, \boldsymbol{\Psi}).
\end{equation}
The matrix \( \boldsymbol{\Sigma} \in \mathbb{R}^{2 \times 2} \) captures the dependence between the outcome and selection equations, while \( \boldsymbol{\Psi} \in \mathbb{R}^{R \times R} \) models the dependence across outcomes. The model defined by \eqref{eq:multi_HS1}--\eqref{eq:multi_HS2} and \eqref{errorM} is referred to as the multiple Heckman selection model with normal errors (MSLn hereafter).

For identifiability, the covariance matrix \( \boldsymbol{\Sigma} \) is parameterized as
\[
\boldsymbol{\Sigma} =
\begin{pmatrix}
\sigma^2 & \rho \sigma \\
\rho \sigma & 1
\end{pmatrix},
\]
where the variance of the selection error is fixed to one.

Ignoring the selection mechanism for the moment, the $n$ independent matrices satisfy
\[
\boldsymbol{\mathcal{Y}}_i \sim \mathcal{N}_{2 \times R}(\mathbf{M}_i = \mathbf{Z}_i \boldsymbol{B}, \boldsymbol{\Sigma}, \boldsymbol{\Psi}), \quad i = 1, \ldots, n.
\]

This matrix-variate formulation provides a convenient framework for modeling multiple outcomes under non-random selection, allowing for a unified treatment of dependence both within and across equations. In practice, the outcome row of the response matrix is observed only when selection occurs (i.e., when \( C_{ri} = 1 \)); otherwise, it is unobserved and treated as missing. The next section formalizes this mechanism by representing the observed data as a censored version of the underlying latent response matrix, following the recent development by \cite{lachos2025}.

\subsection{Censoring structure in the MSLn model}

Let \( \widetilde{\boldsymbol{\mathcal{Y}}} = \{ \boldsymbol{\mathcal{Y}}_1, \ldots, 
\boldsymbol{\mathcal{Y}}_n \} \) denote a sample of latent matrices, where each 
\( \boldsymbol{\mathcal{Y}}_i \in \mathbb{R}^{2 \times R} \). The first row of 
\( \boldsymbol{\mathcal{Y}}_i \) contains the outcome variables, while the second row 
corresponds to latent selection variables associated with the \( R \) responses.

Following \citet{lachos2021heckman}, the observed data for each individual \( i \) is 
represented by the pair \( (\boldsymbol{\mathcal{V}}_i, \mathbf{C}_i) \), where 
\( \boldsymbol{\mathcal{V}}_i \in \mathbb{R}^{2 \times R} \) is a censored version of 
\( \boldsymbol{\mathcal{Y}}_i \), and \( \mathbf{C}_i = (C_{1i}, \ldots, C_{Ri})^{\top} \), 
with \( C_{ri} \in \{0,1\} \) indicating whether the $r$th outcome is observed 
(\( C_{ri} = 1 \)) or unobserved (\( C_{ri} = 0 \)). Formally,
\[
C_{ri} =
\begin{cases}
1, & \text{if } Y_{r2i} > 0, \\[0.5ex]
0, & \text{if } Y_{r2i} \leq 0,
\end{cases}
\]
and \( \boldsymbol{\mathcal{V}}_i \) encodes the observed and censored components of 
\( \boldsymbol{\mathcal{Y}}_i \), with \( \mathcal{V}_{r1i} = Y_{r1i} \) when 
\( C_{ri} = 1 \), and unobserved otherwise.

Under this formulation, the selection variables \( Y_{r2i} \) are subject to truncation 
at zero: they are left-truncated when \( C_{ri} = 1 \) and right-truncated when 
\( C_{ri} = 0 \). In our formulation, the values of \( Y_{r2i} \) are always censored, 
and the unobserved components of $Y_{r1i}$ are treated as censored over the entire real 
line and are handled through integration in the likelihood. This representation allows 
the observed data to be expressed in a unified way, facilitating likelihood-based 
inference via partitioning of the vectorized responses.

{ To characterize the conditional distribution of the observed outcomes under the MSLn 
model (see Section~\ref{sec:matrix_form}), we work with the row-wise vectorization
\[
\mathbf{y}_i = \operatorname{vec}(\boldsymbol{\mathcal{Y}}_i^\top) \in \mathbb{R}^{2R},
\]
so that the outcome and selection components form contiguous blocks of length~$R$, 
and the joint covariance takes the form $\boldsymbol{\Omega} = \boldsymbol{\Sigma} 
\otimes \boldsymbol{\Psi}$. Note that this convention differs from the column-wise 
vectorization $\boldsymbol{\Psi} \otimes \boldsymbol{\Sigma}$ adopted in some 
references only by a permutation of rows and columns induced by a commutation matrix; 
the two representations are algebraically equivalent and yield identical model 
likelihoods. }

We introduce the following additional notation. For individual $i$, let
\[
\mathcal{S}_i = \{\, r \in \{1,\ldots,R\} : C_{ri} = 1 \,\}, 
\qquad R_i^* = |\mathcal{S}_i|,
\]
denote the index set of observed outcomes and its cardinality, respectively. 
Let $\mathbf{y}_{1i}^{\mathrm{obs}}$ and $\mathbf{y}_{2i}^{\mathrm{obs}}$ 
denote the subvectors of $\mathbf{y}_i$ corresponding to outcomes and selection 
variables, restricted to $\mathcal{S}_i$, with corresponding mean vectors
\[
\boldsymbol{\mu}_{1i}^{\mathrm{obs}} 
= \bigl(\mathbf{x}_{ri}^\top\boldsymbol{\beta}_r\bigr)_{r\in\mathcal{S}_i},
\qquad
\boldsymbol{\mu}_{2i}^{\mathrm{obs}} 
= \bigl(\mathbf{w}_{ri}^\top\boldsymbol{\gamma}_r\bigr)_{r\in\mathcal{S}_i},
\]
and let $\boldsymbol{\Psi}^{\mathrm{obs}}_i$ denote the submatrix of $\boldsymbol{\Psi}$ 
indexed by $\mathcal{S}_i$. Conditionally on $\mathbf{C}_i$, the selection subvector 
$\mathbf{y}_{2i}^{\mathrm{obs}}$ is truncated to the positive orthant
\[
\mathcal{A}_i^{\mathrm{obs}}
=
\bigl\{\mathbf{y}_2\in\mathbb{R}^{R_i^*} : y_{2,r}>0,\; r\in\mathcal{S}_i\bigr\}.
\]

The following proposition shows a link between the continuous part
of the MSLn and the SUN distribution.

{\begin{proposition}[SUN representation of the observed outcomes]
\label{prop:SUN}
Under the MSLn model, the conditional distribution of
$\mathbf{y}_{1i}^{\mathrm{obs}} \mid (\mathbf{C}_i, \mathbf{Z}_i)$
belongs to the unified skew-normal (SUN) family
\citep{arellano2006unification}, with parameters
\[
\boldsymbol{\xi}_i = \boldsymbol{\mu}_{1i}^{\mathrm{obs}},
\quad
\boldsymbol{\Omega}_i = \rho\sigma\,\boldsymbol{\Psi}_i^{\mathrm{obs}},
\quad
\boldsymbol{\Delta}_i = (\boldsymbol{\Psi}_i^{\mathrm{obs}})^{1/2},
\quad
\boldsymbol{\tau}_i = \boldsymbol{\mu}_{2i}^{\mathrm{obs}},
\quad
\boldsymbol{\Gamma}_i = \boldsymbol{\Psi}_i^{\mathrm{obs}},
\]
i.e.,
\[
\mathbf{y}_{1i}^{\mathrm{obs}} \mid (\mathbf{C}_i, \mathbf{Z}_i)
\sim 
\mathrm{SUN}_{R_i^\ast, R_i^\ast}(
\boldsymbol{\xi}_i, \boldsymbol{\Omega}_i, \boldsymbol{\Delta}_i,
\boldsymbol{\tau}_i, \boldsymbol{\Gamma}_i),
\]
\end{proposition}      }

\begin{proof}
See Appendix for details.
\end{proof}

{\begin{corollary}[Conditional expectation of observed outcomes]
\label{cor:conditional_expectation_MSLn}
Under the MSLn model, the conditional expectation of the observed outcomes
satisfies
\[
\mathbb{E}\!\left(
\mathbf{y}_{1i}^{\mathrm{obs}}
\mid \mathbf{C}_i,\,\mathbf{Z}_i
\right)
=
\boldsymbol{\mu}_{1i}^{\mathrm{obs}}
+
(\rho\sigma)\,
\boldsymbol{\delta}_i^{\mathrm{obs}}(\mathbf{C}_i),
\]
where
\[
\boldsymbol{\delta}_i^{\mathrm{obs}}(\mathbf{C}_i)
=
\mathbb{E}\!\left(
\mathbf{y}_{2i}^{\mathrm{obs}} - \boldsymbol{\mu}_{2i}^{\mathrm{obs}}
\;\middle|\;
\mathbf{C}_i,\,\mathbf{Z}_i
\right)
\in \mathbb{R}^{R_i^*}
\]
is the truncated-normal mean correction, given explicitly by
\[
\boldsymbol{\delta}_i^{\mathrm{obs}}(\mathbf{C}_i)
=
(\boldsymbol{\Psi}_i^{\mathrm{obs}})^{1/2}\,
\frac{
\phi_{R_i^\ast}\!\left(
(\boldsymbol{\Psi}_i^{\mathrm{obs}})^{-1/2}
\boldsymbol{\mu}_{2i}^{\mathrm{obs}}
\right)
}{
\Phi_{R_i^\ast}\!\left(
(\boldsymbol{\Psi}_i^{\mathrm{obs}})^{-1/2}
\boldsymbol{\mu}_{2i}^{\mathrm{obs}}
\right)
},
\]
which extends the classical inverse Mills ratio to the matrix-variate 
setting.
\end{corollary}  }
\begin{proof}
See Appendix for details.
\end{proof}

{\begin{remark}
As in the classical Heckman model (see Eq.~(\ref{correH})), the conditional 
expectation departs from the marginal mean $\boldsymbol{\mu}_{1i}^{\mathrm{obs}}$ 
by a correction term that reflects the bias induced by nonrandom selection. 
The vector $\boldsymbol{\delta}_i^{\mathrm{obs}}(\mathbf{C}_i)$ extends the 
classical inverse Mills ratio to the matrix-variate setting, where multiple 
outcomes are jointly selected.
\end{remark}
}

{\begin{remark}
The MSLn model nests the classical Heckman selection model as a special 
case. Setting $R = 1$ and $C_i = 1$, so that $R_i^\ast = 1$ and the 
conditional distribution of Proposition~\ref{prop:SUN} reduces to 
$\mathrm{SUN}_{1,1}$, all matrix-variate quantities reduce to scalars 
with $\boldsymbol{\Psi}_i^{\mathrm{obs}} = 1$, since the variance of the 
selection error $\epsilon_{2i}$ is fixed to 1 for identifiability (see 
Eq.~\eqref{nerror} and the discussion therein),
$\boldsymbol{\mu}_{1i}^{\mathrm{obs}} = \mathbf{x}_i^\top\boldsymbol{\beta}$, 
and $\boldsymbol{\mu}_{2i}^{\mathrm{obs}} = \mathbf{w}_i^\top\boldsymbol{\gamma}$.
Moreover, $\mathbf{C}_i$ reduces to the scalar event $\{C_i = 1\}$,
$\mathbf{Z}_i$ contains $(\mathbf{x}_i, \mathbf{w}_i)$, and 
$\mathbf{y}_{1i}^{\mathrm{obs}}$ reduces to the scalar $Y_{1i}$,
so the conditioning in the Corollary becomes $\mid C_i = 1, \mathbf{x}_i, 
\mathbf{w}_i$. Substituting $\boldsymbol{\Psi}_i^{\mathrm{obs}} = 1$ 
directly into the explicit formula for 
$\boldsymbol{\delta}_i^{\mathrm{obs}}(\mathbf{C}_i)$ in the Corollary gives
\[
\boldsymbol{\delta}_i^{\mathrm{obs}}(\mathbf{C}_i)
=
\frac{\phi(\mathbf{w}_i^\top\boldsymbol{\gamma})}
{\Phi(\mathbf{w}_i^\top\boldsymbol{\gamma})}
=
\lambda(\mathbf{w}_i^\top\boldsymbol{\gamma}),
\]
and the conditional expectation of the Corollary reduces exactly to 
Eq.~(\ref{correH}), i.e.,
\[
\mathbb{E}\!\left(
\mathbf{y}_{1i}^{\mathrm{obs}}
\mid C_i = 1, \mathbf{x}_i, \mathbf{w}_i
\right)
=
\mathbb{E}\!\left(
Y_{1i} \mid C_i = 1, \mathbf{x}_i, \mathbf{w}_i
\right)
=
\mathbf{x}_i^\top\boldsymbol{\beta}
+
\rho\sigma\,\lambda(\mathbf{w}_i^\top\boldsymbol{\gamma}),
\]
confirming that the MSLn model nests the classical Heckman selection 
model as a special case.
\end{remark}
}

\subsection{Likelihood function}\label{sec:likelihood}

For computational convenience, we use the vectorized representation of the matrix-variate normal distribution given in \eqref{eq:iff}. Let
\[
\mathbf{y}_i = \mathrm{vec}(\boldsymbol{\mathcal{Y}}_i) \in \mathbb{R}^{2R},
\qquad
\mathbf{v}_i = \mathrm{vec}(\boldsymbol{\mathcal{V}}_i) \in \mathbb{R}^{2R},
\]
and define
\[
\mathrm{vec}(\widetilde{\boldsymbol{\mathcal{Y}}}) = \{\mathbf{y}_1,\ldots,\mathbf{y}_n\},
\qquad
\mathrm{vec}(\widetilde{\boldsymbol{\mathcal{V}}}) = \{\mathbf{v}_1,\ldots,\mathbf{v}_n\},
\qquad
\boldsymbol{C} = \{\boldsymbol{C}_1,\ldots,\boldsymbol{C}_n\}.
\]

For each \(i\), partition \( \mathbf{y}_i \) into observed and censored components as
\[
\mathbf{y}_i =
\begin{pmatrix}
\mathbf{y}_i^o \\
\mathbf{y}_i^c
\end{pmatrix},
\]
where \( \mathbf{y}_i^o \in \mathbb{R}^{p_i^o} \) contains the observed entries and \( \mathbf{y}_i^c \in \mathbb{R}^{p_i^c} \) contains the censored entries, with \( p_i^o + p_i^c = 2R \), according to the selection indicator \( \boldsymbol{C}_i \), \( \ii \).

Let
\[
\boldsymbol{\mu}_i = \mathrm{vec}(\mathbf{M}_i)
= \mathrm{vec}(\mathbf{Z}_i \boldsymbol{B}) \in \mathbb{R}^{2R},
\qquad
\boldsymbol{\Lambda} = \boldsymbol{\Psi} \otimes \boldsymbol{\Sigma} \in \mathbb{R}^{2R \times 2R},
\]
denote the mean vector and covariance matrix of \( \mathbf{y}_i \). Conformably with the partition of \( \mathbf{y}_i \), write
\[
\boldsymbol{\mu}_i =
\begin{pmatrix}
\boldsymbol{\mu}_i^o \\
\boldsymbol{\mu}_i^c
\end{pmatrix},
\qquad
\boldsymbol{\Lambda}_i =
\begin{pmatrix}
\boldsymbol{\Lambda}_i^{oo} & \boldsymbol{\Lambda}_i^{oc} \\
\boldsymbol{\Lambda}_i^{co} & \boldsymbol{\Lambda}_i^{cc}
\end{pmatrix},
\]
where \( \boldsymbol{\Lambda}_i^{oo}, \boldsymbol{\Lambda}_i^{oc}, \boldsymbol{\Lambda}_i^{co} \), and \( \boldsymbol{\Lambda}_i^{cc} \) denote the submatrices of \( \boldsymbol{\Lambda} \) associated with the observation pattern of the \(i\)th individual.

By standard properties of the multivariate normal distribution,
\[
\mathbf{y}_i^o \sim \mathcal{N}_{p_i^o}(\boldsymbol{\mu}_i^o,\boldsymbol{\Lambda}_i^{oo}),
\]
and
\[
\mathbf{y}_i^c \mid \mathbf{y}_i^o = \mathbf{y}_i^o
\sim
\mathcal{N}_{p_i^c}(\boldsymbol{\mu}_i^{co},\boldsymbol{\Lambda}_i^{cc.o}),
\]
where
\begin{equation}
\boldsymbol{\mu}_i^{co}
=
\boldsymbol{\mu}_i^c
+
\boldsymbol{\Lambda}_i^{co}
(\boldsymbol{\Lambda}_i^{oo})^{-1}
(\mathbf{y}_i^o-\boldsymbol{\mu}_i^o),
\qquad
\boldsymbol{\Lambda}_i^{cc.o}
=
\boldsymbol{\Lambda}_i^{cc}
-
\boldsymbol{\Lambda}_i^{co}
(\boldsymbol{\Lambda}_i^{oo})^{-1}
\boldsymbol{\Lambda}_i^{oc}.
\label{eqn mu.co S.co}
\end{equation}

Following \citet{lachos2025}, the observed-data likelihood for
\[
\boldsymbol{\theta}
=
(\boldsymbol{\beta}_1,\ldots,\boldsymbol{\beta}_R,
\boldsymbol{\gamma}_1,\ldots,\boldsymbol{\gamma}_R,
\boldsymbol{\Sigma},\boldsymbol{\Psi})
\]
is given by
\begin{equation}
L(\boldsymbol{\theta}\mid \mathrm{vec}(\widetilde{\boldsymbol{\mathcal{V}}}),\boldsymbol{C})
=
\prod_{i=1}^n L_i(\boldsymbol{\theta}),
\label{eq:obslik}
\end{equation}
where the contribution of the \(i\)th observation is
\begin{equation}
L_i(\boldsymbol{\theta})
=
P\!\left(
\mathbf{v}_{1i}^c \le \mathbf{y}_i^c \le \mathbf{v}_{2i}^c
\,\middle|\,
\mathbf{y}_i^o,\boldsymbol{\theta}
\right)
f(\mathbf{y}_i^o\mid\boldsymbol{\theta}).
\label{eq:Li}
\end{equation}
The integration bounds \(\mathbf{v}_{1i}^c\) and \(\mathbf{v}_{2i}^c\) encode the censoring structure induced by \(\boldsymbol{C}_i\). Specifically, for each \(r\), the latent selection variable \(Y_{r2i}\) is constrained to \((0,+\infty)\) when \(C_{ri}=1\) and to \((-\infty,0]\) when \(C_{ri}=0\), whereas unobserved outcome components \(Y_{r1i}\) are left unrestricted over \((-\infty,+\infty)\).
Using the conditional and marginal normal distributions above, we obtain
\begin{equation}
L_i(\boldsymbol{\theta})
=
\Phi_{p_i^c}\!\left(
\mathbf{v}_{1i}^c,\mathbf{v}_{2i}^c
\mid
\boldsymbol{\mu}_i^{co},\boldsymbol{\Lambda}_i^{cc.o}
\right)
\,
\phi_{p_i^o}\!\left(
\mathbf{y}_i^o;
\boldsymbol{\mu}_i^o,\boldsymbol{\Lambda}_i^{oo}
\right),
\label{eq:Li_closed}
\end{equation}
where \( \phi_{p_i^o}(\cdot;\boldsymbol{\mu}_i^o,\boldsymbol{\Lambda}_i^{oo}) \) denotes the \(p_i^o\)-variate normal density and
\( \Phi_{p_i^c}(\mathbf{v}_{1i}^c,\mathbf{v}_{2i}^c \mid \boldsymbol{\mu}_i^{co},\boldsymbol{\Lambda}_i^{cc.o}) \) denotes the probability that a \(p_i^c\)-variate normal random vector with mean \( \boldsymbol{\mu}_i^{co} \) and covariance matrix \( \boldsymbol{\Lambda}_i^{cc.o} \) falls in the rectangle \( [\mathbf{v}_{1i}^c,\mathbf{v}_{2i}^c] \).

Therefore, the observed-data log-likelihood is
\begin{equation}
\ell(\boldsymbol{\theta}\mid \mathrm{vec}(\widetilde{\boldsymbol{\mathcal{V}}}),\boldsymbol{C})
=
\sum_{i=1}^n \log L_i(\boldsymbol{\theta})
=
\sum_{i=1}^n
\log
\Phi_{p_i^c}\!\left(
\mathbf{v}_{1i}^c,\mathbf{v}_{2i}^c
\mid
\boldsymbol{\mu}_i^{co},\boldsymbol{\Lambda}_i^{cc.o}
\right)
+
\sum_{i=1}^n
\log
\phi_{p_i^o}\!\left(
\mathbf{y}_i^o;
\boldsymbol{\mu}_i^o,\boldsymbol{\Lambda}_i^{oo}
\right).
\label{eq:loglik_obs}
\end{equation}

Direct maximization of \eqref{eq:loglik_obs} is computationally challenging because the likelihood involves multivariate normal rectangle probabilities whose dimension and partition structure vary across observations according to the censoring pattern. To address this difficulty, we develop an Expectation--Conditional Maximization (ECM) algorithm based on the complete-data representation of the latent responses. This approach replaces direct maximization of the observed-data likelihood by a sequence of simpler conditional maximization steps.

\section{The ECM algorithm}\label{sec:ecm_algorithm}

The ECM algorithm, proposed by \citet{Meng93}, replaces the M-step of the EM algorithm with a series of simpler conditional maximization (CM) steps. This modification retains EM's stability while often improving convergence speed \citep{McLachlanKrishnan}. The parameter updates in our ECM algorithm closely follow the framework established by \citet{lachos2025}, who proposed a matrix-variate EM-type approach for incomplete and censored data. Given the structural similarity between their setting and the matrix-variate selection model considered here, we adopt their estimation strategy with notational and contextual adjustments. To avoid redundancy, we omit the full derivations and present only the resulting update equations, referring the interested reader to \citet{lachos2025} for technical details and proofs.

In the complete-data formulation, we augment the observed matrices \(\mathcal{V}_i\), selection indicators \(\bC_i\), and censored values with the full latent matrices \(\boldsymbol{\mathcal{Y}_i} \in \mathbb{R}^{2 \times R}\), forming the augmented dataset \(\mathcal{S} = \{(\boldsymbol{\mathcal{Y}_i}, \boldsymbol{\mathcal{V}_i}, \bC_i): i = 1, \dots, n\}\).

The complete-data log-likelihood is given by
\[
\ell_c(\boldsymbol{\theta}) = -nR \log(2\pi) - \frac{nR}{2} \log|\boldsymbol{\Sigma}| - n \log|\boldsymbol{\Psi}| 
- \frac{1}{2} \sum_{i=1}^n \mathrm{tr}\left[ \boldsymbol{\Sigma}^{-1} (\boldsymbol{{\mathcal{Y}}_i} - \mathbf{M}_i) \boldsymbol{\Psi}^{-1} (\boldsymbol{{\mathcal{Y}}_i} - \mathbf{M}_i)^\top \right],
\]
where $\btheta=(\bbeta_1,\ldots,\bbeta_R,\bgamma_1,\ldots,\bgamma_R,\bSigma,\bPsi)=(\bbeta^c_1,\ldots,\bbeta^c_R,\bSigma,\bPsi)$, with $\bbeta_r^c$ as defined in (\ref{betac}) and $\mathbf{M}_i={\mathbf{Z}}_i {\boldsymbol{B}}$.
Subsequently, the ECM algorithm for the MSLn model can be summarized as follows:

\indent{\em \bf E-step:}
Given the current estimate  $\widehat{\btheta}^{\ka} = \{\widehat{\bbeta}^{c\ka}_1,\ldots, \widehat{\bbeta}^{c\ka}_R, \widehat{\bSigma}^{\ka}, \widehat{\bPsi}^{\ka}\}$ at the $k$th step of the algorithm, the E-step provides the conditional expectation of the complete data log-likelihood function, i.e.,
\begin{eqnarray} Q(\btheta\mid\widehat{\btheta}^{\ka})=\mathbb{E} \Bigl[ \ell_c(\btheta)\mid\widetilde{\boldsymbol{\mathcal{V}}},{{\bC}},\widehat{\btheta}^{\ka} \Bigr] = \sumas{Q_i(\btheta\mid\widehat{\btheta}^{\ka})},\label{Q1-1} \end{eqnarray}
where
\begin{eqnarray*}
	Q_i(\btheta\mid\widehat{\btheta}^{\ka})&=& - R \log(2\pi) -\log |\bPsi| - \frac{R}{2} \log|\bSigma|-\frac{1}{2}\tr\left\{\left(\bPsi\otimes\bSigma\right)^{-1} \widehat{\bDelta}^{(k)}\right\},
\end{eqnarray*}
with $$\widehat{\bDelta}^{(k)}=\left[\left(\text{vec}(\widehat{\boldsymbol{\mathcal{Y}}}^{\ka}_i)-\text{vec}({\mathbf{Z}}_i {\boldsymbol{B}})\right)\left(\text{vec}(\widehat{\boldsymbol{\mathcal{Y}}}^{\ka}_i)-\text{vec}({\mathbf{Z}}_i {\boldsymbol{B}})\right)^{\top}+ \widehat{\bV}^{\ka}_i\right],$$ 
where $\text{vec}(\widehat{\boldsymbol{\mathcal{Y}}}_i^{\ka})= \mathbb{E}_{\boldsymbol{\mathcal{Y}}_i}[\text{vec}(\boldsymbol{\mathcal{Y}}_i) \,|\, \widetilde{\boldsymbol{\mathcal{V}}}_i, \bC_i, \widehat{\btheta}^{\ka}]$ and
$\widehat{\bV}^{\ka}_i={\mathbb Cov}_{\boldsymbol{\mathcal{Y}}_i}[\text{vec}(\boldsymbol{\mathcal{Y}}_i) \,|\, \widetilde{\boldsymbol{\mathcal{V}}}_i, {\bC}_i, \widehat{\btheta}^{\ka}]$.\\

The E-step reduces only to the computation of $\text{vec}(\widehat{\boldsymbol{\mathcal{Y}}}_i^{\ka})$ and $\widehat{\bV}^{\ka}_i$, that is, the first and second moments of a truncated multinormal distribution. These can be determined in closed form, as a function of MN probabilities using a sequence of simple transformations, for which we use the \texttt{MomTrunc} package in \texttt{R} \citep{GalarzaCran}. 

Before discussing the implementation of the CM-steps, the following results are essential. For clarity of exposition, we omit the superscript $k$ and subscript $i$ in the next lemma.

{
\begin{Lemma}\label{lemma1}
Under the MSLn model, the matrix $\widehat{\bDelta}$ of dimension $2R\times 2R$, which appears in the $Q$-function (\ref{Q1-1}), is symmetric and positive definite $(\succ)$.
\end{Lemma}
}

\begin{proof}
See Appendix for details.
\end{proof}

{Lemma~\ref{lemma1} shows that the matrix
$\widehat{\boldsymbol{\Delta}}$ is symmetric and positive definite.
This property guarantees that the quadratic term in the $Q$-function
is properly defined and induces a strictly convex optimization problem.
\\
From a computational perspective, this ensures that the ECM updates are stable and that the parameter estimates are uniquely determined at
each iteration. Hence, the result provides a key theoretical justification
for the convergence behavior of the proposed algorithm.}\\

\indent{\bf \em CM-step 1:} Updating equation for $\boldsymbol{B}$.\\

For fixed $\widehat{\bSigma}^{\ka}$ and $\widehat{\bPsi}^{\ka}$, the first CM-step consists of maximizing 
$Q(\btheta\mid\widehat{\btheta}^{\ka})$ with respect to $\boldsymbol{B}$.  From (\ref{Q1-1}), and using the identity 
$\text{vec}(\mathbf{A})^\top (\bPsi \otimes \bSigma)^{-1} \text{vec}(\mathbf{A})
= \tr\big(\bSigma^{-1} \mathbf{A} \bPsi^{-1} \mathbf{A}^\top\big)$, 
the part of the $Q$-function that depends on $\boldsymbol{B}$ reduces, up to an additive constant, to
\[
-\frac{1}{2}\sum_{i=1}^n 
\tr\left[
\widehat{\bSigma}^{-1\ka}
\left(\boldsymbol{\widehat{\mathcal{Y}}}_i^{\ka}-\mathbf{Z}_i\boldsymbol{B}\right)
\widehat{\bPsi}^{-1\ka}
\left(\boldsymbol{\widehat{\mathcal{Y}}}_i^{\ka}-\mathbf{Z}_i\boldsymbol{B}\right)^\top
\right].
\]

Differentiating with respect to $\boldsymbol{B}$ and setting the derivative equal to zero, we obtain
\[
\sum_{i=1}^n
\mathbf{Z}_i^\top
\widehat{\bSigma}^{-1\ka}
\left(\boldsymbol{\widehat{\mathcal{Y}}}_i^{\ka}-\mathbf{Z}_i\boldsymbol{B}\right)
\widehat{\bPsi}^{-1\ka}
=\mathbf{0}.
\]

Hence,
\[
\left(
\sum_{i=1}^n
\mathbf{Z}_i^\top \widehat{\bSigma}^{-1\ka}\mathbf{Z}_i
\right)\boldsymbol{B}\,\widehat{\bPsi}^{-1\ka}
=
\sum_{i=1}^n
\mathbf{Z}_i^\top \widehat{\bSigma}^{-1\ka}\boldsymbol{\widehat{\mathcal{Y}}}_i^{\ka}\,\widehat{\bPsi}^{-1\ka}.
\]

Since $\widehat{\bPsi}^{\ka}$ is nonsingular, right-multiplication by $\widehat{\bPsi}^{\ka}$ yields
\begin{equation}\label{eq:Bupdate}
\widehat{\boldsymbol{B}}^{\ke}
=
\left(
\sum_{i=1}^n
\mathbf{Z}_i^\top \widehat{\bSigma}^{-1\ka}\mathbf{Z}_i
\right)^{-1}
\left(
\sum_{i=1}^n
\mathbf{Z}_i^\top \widehat{\bSigma}^{-1\ka}\boldsymbol{\widehat{\mathcal{Y}}}_i^{\ka}
\right).
\end{equation}

Note that $\boldsymbol{B}$ is block diagonal, with $r$th block $\boldsymbol{\beta}_r^c$ as defined in \eqref{betac}. Since $\mathbf{Z}_i$ shares the same structure, the minimization decouples over $r=1,\ldots,R$. Let $\widehat{\by}_{ir}^{\ka}$ denote the $r$th column of $\boldsymbol{\widehat{\mathcal{Y}}_i^{\ka}}$ and define
\[
\mathbf{Z}_{ir}
=
\begin{pmatrix}
\mathbf{x}_{ir}^\top & \mathbf{0}^\top\\
\mathbf{0}^\top & \mathbf{w}_{ir}^\top
\end{pmatrix}.
\]
Then
\begin{equation}\label{eq:betac_update}
\widehat{\boldsymbol{\beta}}_r^{c\ke}
=
\left(
\sum_{i=1}^n
\mathbf{Z}_{ir}^\top \widehat{\bSigma}^{-1\ka}\mathbf{Z}_{ir}
\right)^{-1}
\left(
\sum_{i=1}^n
\mathbf{Z}_{ir}^\top \widehat{\bSigma}^{-1\ka}\widehat{\by}_{ir}^{\ka}
\right),
\qquad r=1,\ldots,R.
\end{equation}

Writing
\[
\widehat{\bSigma}^{-1\ka}
=
\frac{1}{\widehat{\sigma}^{2\ka}(1-\widehat{\rho}^{2\ka})}
\begin{pmatrix}
1 & -\widehat{\rho}^{\ka}\widehat{\sigma}^{\ka}\\
-\widehat{\rho}^{\ka}\widehat{\sigma}^{\ka} & \widehat{\sigma}^{2\ka}
\end{pmatrix},
\]
and letting $\widehat{\by}_{ir}^{\ka}=(\widehat{Y}_{1ri}^{\ka},\widehat{Y}_{2ri}^{\ka})^\top$, we obtain
\[
\widehat{\boldsymbol{\beta}}_r^{c\ke}
=
\begin{pmatrix}
\widehat{\boldsymbol{\beta}}_r^{\ke}\\[0.3em]
\widehat{\boldsymbol{\gamma}}_r^{\ke}
\end{pmatrix},
\]
with
\begin{align}
\widehat{\boldsymbol{\beta}}_r^{\ke}
&=
\left(
\sum_{i=1}^n
\frac{\mathbf{x}_{ir}\mathbf{x}_{ir}^\top}
{\widehat{\sigma}^{2\ka}(1-\widehat{\rho}^{2\ka})}
\right)^{-1}
\left(
\sum_{i=1}^n
\frac{\mathbf{x}_{ir}\left(\widehat{Y}_{1ri}^{\ka}
-\widehat{\rho}^{\ka}\widehat{\sigma}^{\ka}\widehat{Y}_{2ri}^{\ka}\right)}
{\widehat{\sigma}^{2\ka}(1-\widehat{\rho}^{2\ka})}
\right),
\label{eq:beta_update}
\\[1ex]
\widehat{\boldsymbol{\gamma}}_r^{\ke}
&=
\left(
\sum_{i=1}^n
\frac{\mathbf{w}_{ir}\mathbf{w}_{ir}^\top}
{1-\widehat{\rho}^{2\ka}}
\right)^{-1}
\left(
\sum_{i=1}^n
\frac{\mathbf{w}_{ir}\left(\widehat{Y}_{2ri}^{\ka}
-\widehat{\rho}^{\ka}\widehat{\sigma}^{-\ka}\widehat{Y}_{1ri}^{\ka}\right)}
{1-\widehat{\rho}^{2\ka}}
\right).
\label{eq:gamma_update}
\end{align}

After obtaining $\widehat{\boldsymbol{B}}^{\ke}$, define
\[
\widehat{\bDelta}^{\ka *}_i
=
\left[
\left(\text{vec}(\boldsymbol{\widehat{\mathcal{Y}}}^{\ka}_i)-\text{vec}({\mathbf{Z}}_i \widehat{\boldsymbol{B}}^{\ke})\right)
\left(\text{vec}(\boldsymbol{\widehat{\mathcal{Y}}}^{\ka}_i)-\text{vec}({\mathbf{Z}}_i \widehat{\boldsymbol{B}}^{\ke})\right)^{\top}
+ \widehat{\bV}^{\ka}_i
\right].
\]
By Lemma \ref{lemma1}, $\widehat{\bDelta}^{\ka *}_i$ is symmetric positive definite, and therefore admits a Cholesky decomposition. This property allows the $Q$-function, evaluated at $\widehat{\boldsymbol{B}}^{\ke}$, to be rewritten in a form depending only on $\bSigma$ and $\bPsi$, as established in the following theorem.

In the following Theorem, we use the notation $\widehat{\boldsymbol{B}}^{(k)}$ to represent

\[
\widehat{\boldsymbol{B}}^{(k)} =
\begin{pmatrix}
\widehat{\boldsymbol{\beta}}^{c(k)}_1 & \mathbf{0} & \cdots & \mathbf{0} \\
\mathbf{0} & \widehat{\boldsymbol{\beta}}^{c(k)}_2 & \cdots & \mathbf{0} \\
\vdots & \vdots & \ddots & \vdots \\
\mathbf{0} & \mathbf{0} & \cdots & \widehat{\boldsymbol{\beta}}^{c(k)}_R
\end{pmatrix}.
\]

{
\begin{theorem} At  $k$th step of the EM algorithm the $Q$-function in (\ref{Q1-1}), with $\boldsymbol{B}$ replaced by $\widehat{\boldsymbol{B}}^{(k+1)}$, can be expressed as
  \begin{equation}\label{Q-new}
  Q(\widehat{\boldsymbol{B}}^{\ke}, {\bSigma}, {\bPsi}\mid\widehat{\btheta}^{\ka})=  - n \log |\bPsi| - \frac{nR}{2} \log|\bSigma|-\frac{1}{2}\sumas\sum_{j=1}^{2R}\tr \left[\bSigma^{-1}\widehat{\bD}^{\ka}_{ij}\bPsi^{-1}\widehat{\bD}^{\top\ka}_{ij}\right],
  \end{equation}
  where $\widehat{\bD}^{\ka}_{ij}$ is a $2\times R$ matrix such that $\operatorname{vec}(\widehat{\bD}^{\ka}_{ij})=\widehat{\LL}^{\ka}_{ij}$.
Here $\widehat{\LL}^{\ka}_{ij}$ is $j$th column of the $2R\times 2R$ lower triangular matrix $\widehat{\LL}^{\ka}_{i}$, obtained from the Cholesky decomposition of the matrix 
\begin{align*}
    \widehat{\bDelta}^{\ka*}_i &= \left[\left(\text{vec}(\boldsymbol{\widehat{\mathcal{Y}}}^{\ka}_i)-\text{vec}({\mathbf{Z}}_i \widehat{\boldsymbol{B}}^{(k+1)})\right)\left(\text{vec}(\boldsymbol{\widehat{\mathcal{Y}}}^{\ka}_i)-\text{vec}({\mathbf{Z}}_i \widehat{\boldsymbol{B}}^{(k+1)})\right)^{\top}+ \widehat{\bV}^{\ka}_i\right].
\end{align*}
\end{theorem}
}
\begin{proof}
See Appendix for details.
\end{proof}
{
Theorem~1 shows that the $Q$-function admits a tractable
representation based on the Cholesky decomposition of
$\widehat{\boldsymbol{\Delta}}^{(k)*}_i$. By exploiting this
structure, the objective function can be expressed as a sum of
trace terms involving $\boldsymbol{\Sigma}^{-1}$ and
$\boldsymbol{\Psi}^{-1}$.
\\
This result is crucial for computational efficiency,
as it enables a decoupled and numerically stable
implementation of the M-step.}\\

\indent{\bf \em CM-step 2:}  Updating equations for $\bSigma$ and $\bPsi$.\\

Conditional maximization of $Q({\widehat{\boldsymbol{B}}}^{\ke}, {\bSigma}, {\bPsi}\mid\widehat{\btheta}^{\ka})$ in Theorem 1, with respect to $\bSigma$ and $\bPsi$ \citep[see,][equations 3 and 4]{glanz2018expectation}, yields the following update equations:
\begin{align}
\widehat{\bSigma}^{\ke} &= \frac{1}{nR}\sum_{i=1}^{n}\sum_{j=1}^{2R}\widehat{\bD}^{\ka}_{ij}\widehat{\bPsi}^{{\ka}{^{-1}}} \widehat{\bD}^{\ka\top}_{ij},\label{eq:Sigma}\\
\widehat{\bPsi}^{\ke} &= \frac{1}{2n}\sum_{i=1}^{n}\sum_{j=1}^{2R}\widehat{\bD}^{\ka\top}_{ij}\widehat{\bSigma}^{{\ka}{^{-1}}} \widehat{\bD}^{\ka}_{ij}.
\label{eq:Psi}
\end{align}	
Note that, to satisfy the identifiability constraint on $\bSigma$, we replace its second diagonal element by $\widehat{\bSigma}^{\ke} = 1$ at each iteration of the algorithm.

The ML estimation method via the EM algorithm may not yield global solutions if the initial values are far from the true parameter values. Thus, the choice of starting values plays an important role in parameter estimation. In our examples and simulation studies, we consider the following procedure for the MSLn model: $(i)$ The average of the observed values replaces missing values and censored observations to obtain complete data. $(ii)$ With the complete data, we use the \texttt{R} library \texttt{MixMatrix} to obtain initial values for $\bPsi^{(0)}$, and $\bSigma^{(0)}$. $(iii)$ For the initial value of the regression parameters, we use the \texttt{R} library \texttt{HeckmanEM} for each outcome.

Finally, we stopped the algorithm when $|\ell(\widehat{\btheta}^{\ke} | \text{vec}(\widetilde{\mathcal{V}}),\text{vec}(\widetilde{\mathcal{C}}))/\ell(\widehat{\btheta}^{\ka} | \text{vec}(\widetilde{\mathcal{V}}),\text{vec}(\widetilde{\mathcal{C}})) - 1| < \epsilon$, with $\epsilon = 10^{-6}$, i.e., the algorithm stops when the relative distance between two successive evaluations of the log-likelihood is less than the tolerance.

\section{Simulation Study}

In this section, we conduct a simulation study under various configurations to evaluate the performance of the proposed Heckman selection model for multivariate outcomes. 
First, we investigate the finite sample properties of the EM estimates under various sample sizes and missing rates. The goal is to assess the accuracy of the parameter estimates, particularly the regression coefficients matrix $\widetilde{\boldsymbol{B}} = (\boldsymbol{B}, \boldsymbol{\Gamma})' 
\in \mathbb{R}^{(p+q) \times R}$, the covariance matrices \(\boldsymbol{\Sigma} \in \mathbb{R}^{2 \times 2}\) and \(\boldsymbol{\Psi} \in \mathbb{R}^{R \times R}\) . Second, we compare the proposed method with the univariate Heckman model \citep{lachos2021heckman}, applied separately to each outcome using the \texttt{HeckmanEM} package. 
The proposed multivariate ECM estimator is implemented in the \texttt{mvHeckman} package (\url{https://github.com/heeju-lim/mvHeckman}).

\subsection{Finite sample properties}

The objective of this simulation is to evaluate how accurately the proposed ECM estimator, implemented in the \texttt{mvHeckman} package, recovers the model parameters in a multivariate Heckman selection setting. We consider the following scenarios: sample sizes \(n \in \{100, 200, 300\}\) and missing rates \(\in \{0.1, 0.25, 0.5\}\).For each individual $i = 1, \ldots, n$ and outcome $r = 1, \ldots ,R $, we define $\mathbf{w}_{ri}^\top = (1, w_{ri1}, w_{ri2})$ and impose the exclusion restriction $\mathbf{x}_{ri}^\top = (1, w_{ri1})$.

\paragraph{Scenario 1 (Baseline covariance structure).}~\\
We consider three distinct outcome designs to allow for outcome-specific covariate distributions while adopting a homogeneous compound-symmetric covariance structure for the outcome errors. In the first design, $w_{1i1} \sim \mathcal{N}_1(0,1)$ and $w_{1i2}\sim \mathcal{N}_1(0,1)$; in the second, $w_{2i1}\sim \mathcal{N}_1(0,1)$ and $w_{2i2}\sim t_{6}$, where $t_{6}$ denotes the Student’s t-distribution with 6 degrees of freedom; and in the third, $w_{3i1}\sim \mathrm{Unif}(-1,1)$, where $\mathrm{Unif}(a,b)$ denotes the uniform distribution on $[a,b]$ and $w_{3i2}\sim \mathcal{N}_1(0,1)$. The outcome correlation adopts a compound-symmetric structure 
\[
\boldsymbol{\Psi}=(1-\phi){\bf I}_q+\phi\,\mathbf{1}\mathbf{1}^\top,\quad \phi=0.4,
\]
for \(q = 3\) outcomes. The common scale is $\sigma=2$ (so $\sigma^2=4$), and the correlation between the outcome and selection errors is $\rho=0.6$. These values induce nontrivial cross-outcome dependence and a moderate selection effect. The true parameter matrices are specified as follows:
%\VHR{HEEju: we are using this notation $\mathcal{N}_p$ for normal!!!!}

\[
\boldsymbol{B} =
\begin{pmatrix}
1 & 1 & 1\\
0.3 & -0.8 & 2.0
\end{pmatrix},~~
\quad
\boldsymbol{\Gamma} =
\begin{pmatrix}
1 & 1 & 1\\
0.3 & -0.5 & 0.2 \\
-0.7 & -1 & 0.6
\end{pmatrix},
\]
\[
\quad
\boldsymbol{\Sigma} = 
\begin{pmatrix}
4 & 1.2 \\
1.2 & 1
\end{pmatrix},~~
\quad
\boldsymbol{\Psi} =
\begin{pmatrix}
1 &  0.4 & 0.4 \\
0.4 &  1 & 0.4 \\
0.4 &  0.4 & 1 
\end{pmatrix}.
\]

To evaluate the performance of the proposed model, we generate 100 Monte Carlo replications under these settings. 
Estimation accuracy for the regression coefficient matrices $\boldsymbol{B}$ and $\boldsymbol{\Gamma}$ is measured using the Frobenius norm error:
\[
\widetilde{\boldsymbol{B}}=\|\widehat{\boldsymbol{B}} - \boldsymbol{B}\|_F,
\quad
\widetilde{\boldsymbol{\Gamma}}=\|\widehat{\boldsymbol{\Gamma}} - \boldsymbol{\Gamma}\|_F.
\]
For the covariance parameters, including $\sigma$, $\rho$, and $\phi$, performance is evaluated based on the mean squared error (MSE) between the estimated and true values.

Figure~\ref{fig:finite_sample1} presents boxplots summarizing the finite sample properties under Scenario~1. As expected, estimation errors decrease with increasing sample size and increase with higher missing rates. These trends confirm that the proposed model consistently recovers the true parameters under moderate correlation and homogeneity in the outcome error structure.

\begin{figure}[htb!]
  \centering
  
  % (a) 위쪽
  \begin{subfigure}[b]{\textwidth}
    \centering
    \includegraphics[width=0.33\textwidth]{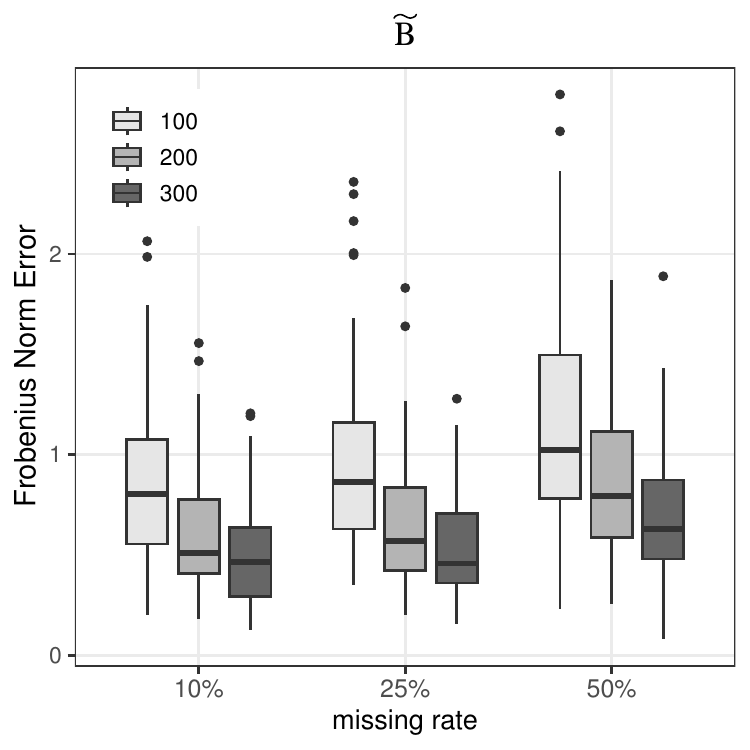}
    \includegraphics[width=0.33\textwidth]{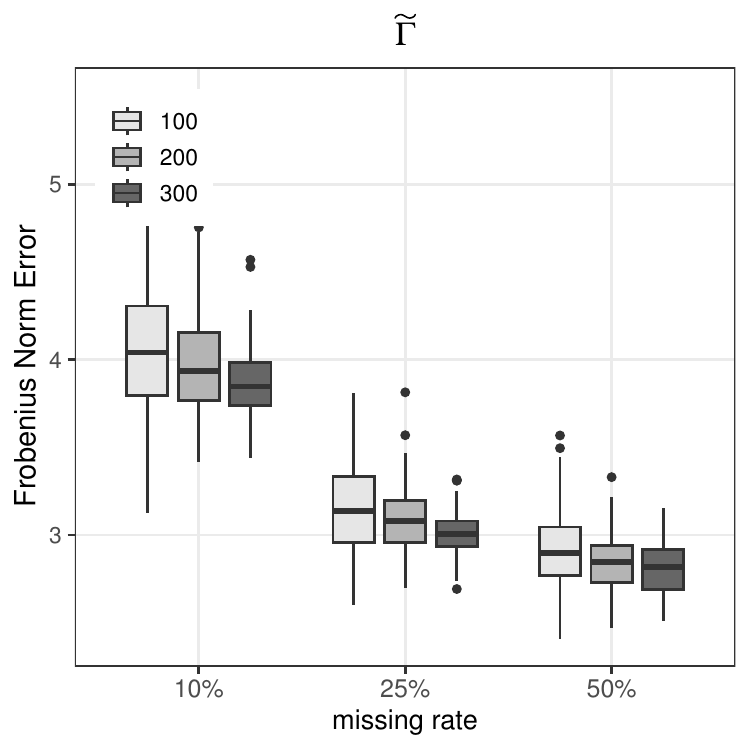}
    \caption{Boxplots of Frobenius norm errors for $\boldsymbol{B}$ and $\boldsymbol{\Gamma}$ based on 100 Monte Carlo replications across sample size and missing rate levels.}
    \label{fig:F_norm}
  \end{subfigure} 
  \vspace{0.5cm} % 위아래 간격 조절  
  % (b) 아래쪽
  \begin{subfigure}[b]{\textwidth}
    \centering
    \includegraphics[width=0.32\textwidth]{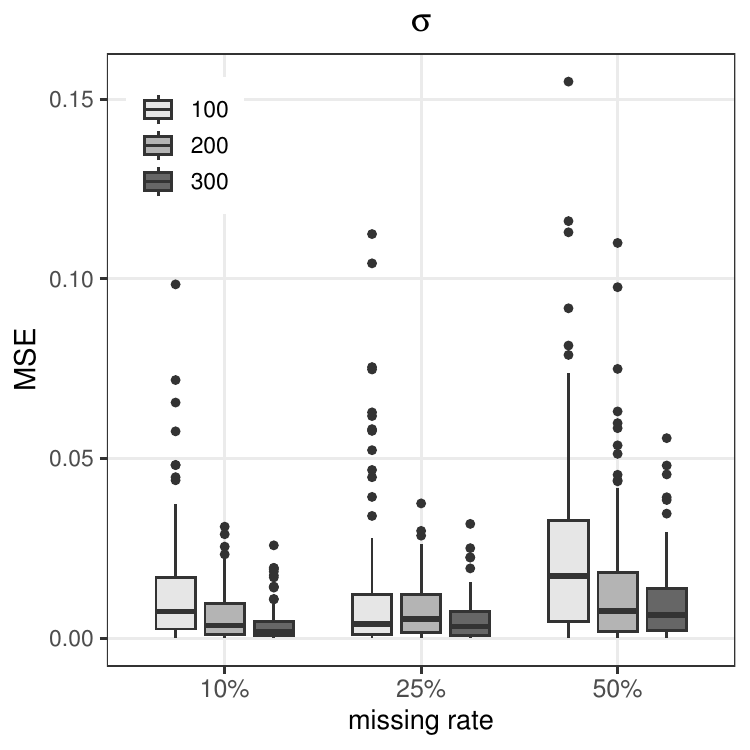}
    \includegraphics[width=0.32\textwidth]{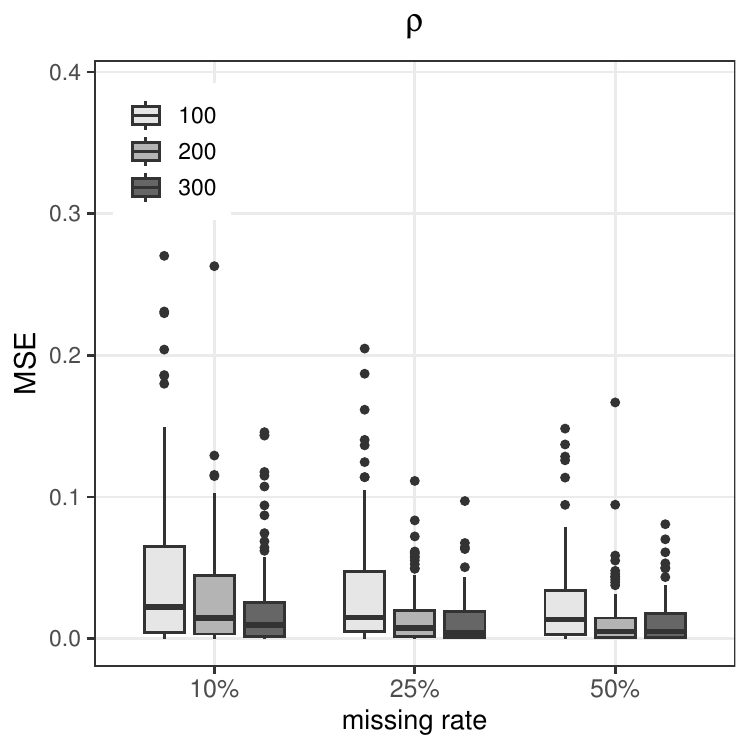}
    \includegraphics[width=0.32\textwidth]{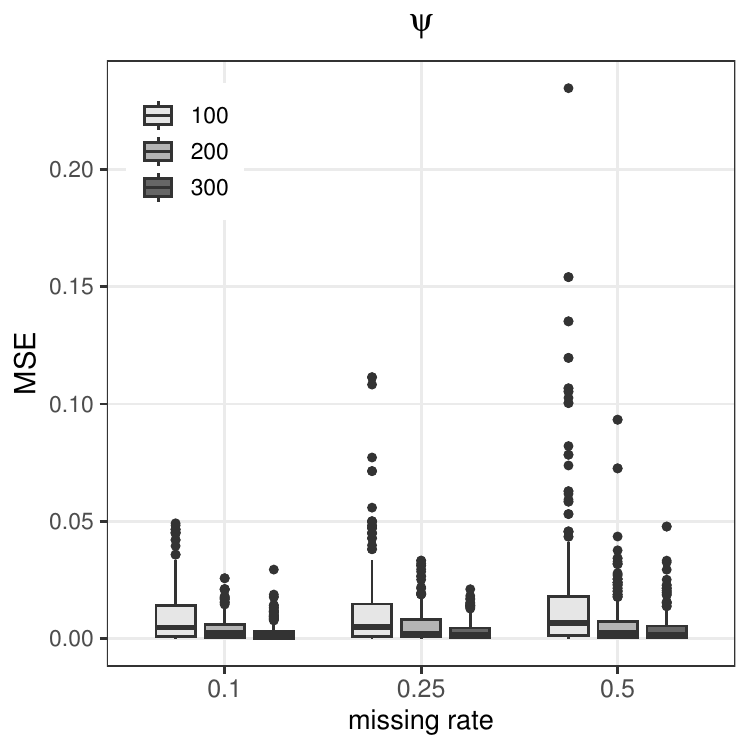}
    \caption{Boxplots of the MSE for $\sigma$, $\rho$, and $\phi$ based on 100 Monte Carlo replications across sample size and missing rate levels.}
  \end{subfigure}
\caption{Scenario~1. Assessment of finite sample properties for each parameter (baseline covariance structure).}
  \label{fig:finite_sample1}
\end{figure}

\begin{figure}[htb!]
  \centering
  
  % (a) 위쪽
  \begin{subfigure}[b]{\textwidth}
    \centering
    \includegraphics[width=0.33\textwidth]{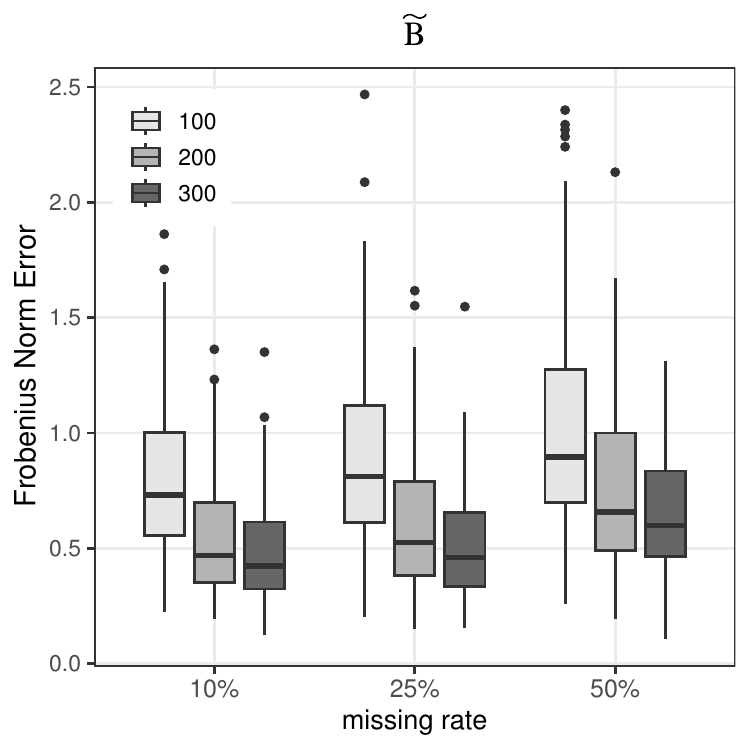}
    \includegraphics[width=0.33\textwidth]{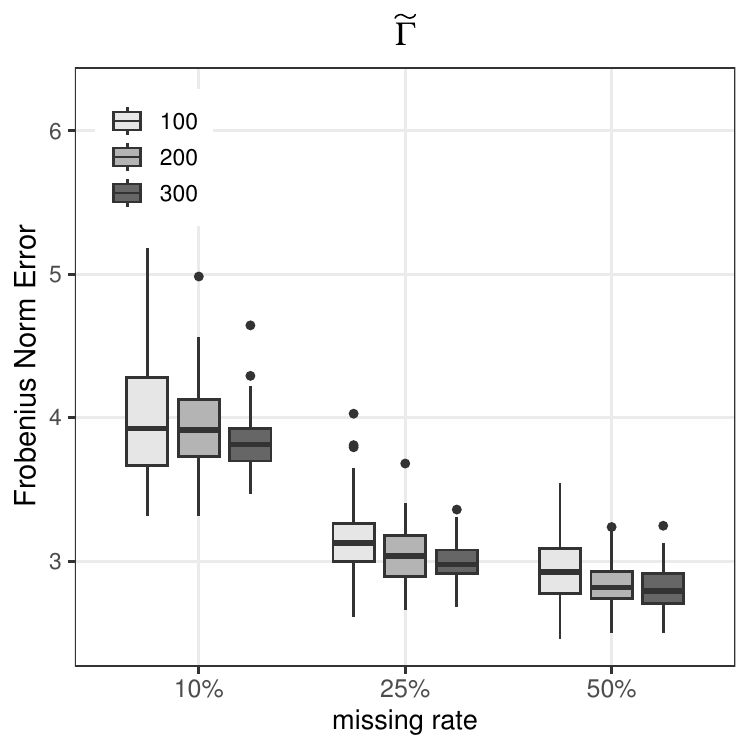}
    \caption{Boxplots of Frobenius norm errors for $\boldsymbol{B}$ and $\boldsymbol{\Gamma}$ based on 100 Monte Carlo replications across sample size and missing rate levels.}
  \end{subfigure} 
  \vspace{0.5cm} % 위아래 간격 조절  
  % (b) 아래쪽
  \begin{subfigure}[b]{\textwidth}
    \centering
    \includegraphics[width=0.32\textwidth]{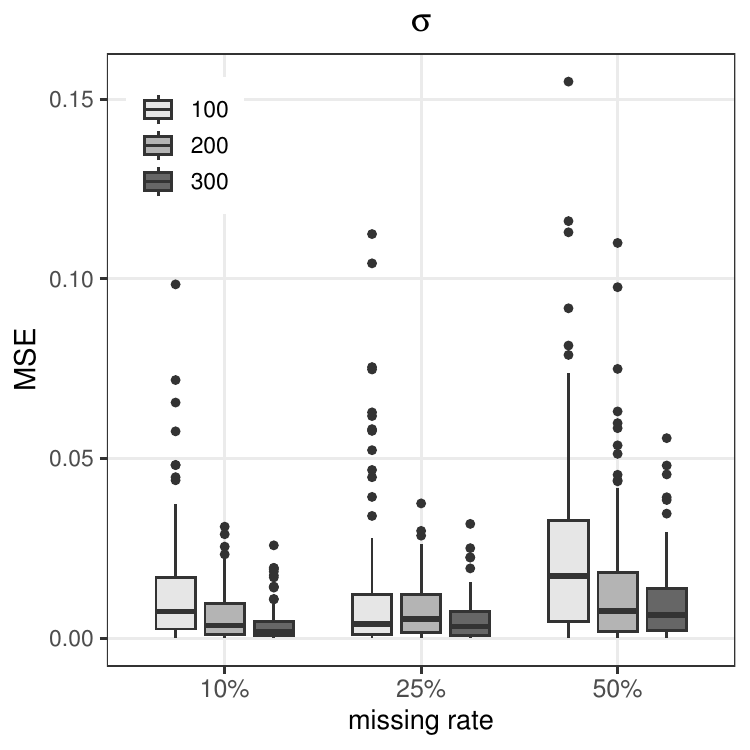}
    \includegraphics[width=0.32\textwidth]{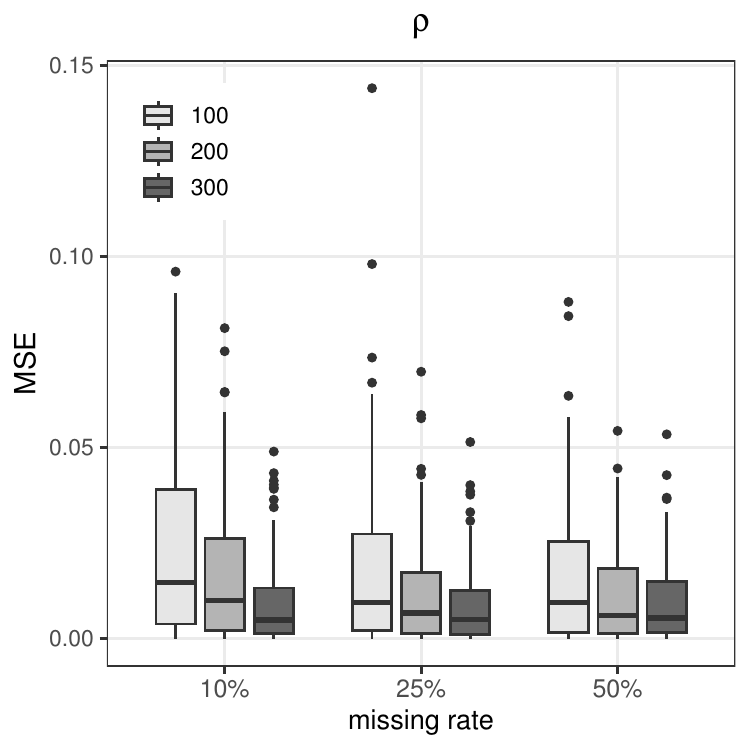}
    \includegraphics[width=0.32\textwidth]{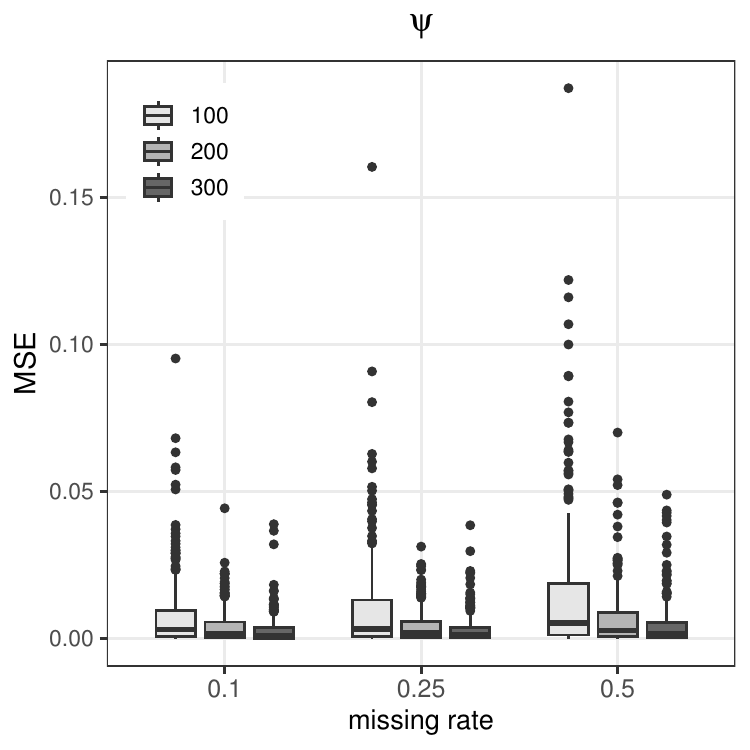}
    \caption{Boxplots of the MSE for $\sigma$, $\rho$, and $\phi$ based on 100 Monte Carlo replications across sample size and missing rate levels.}
    \label{fig:MSE}
  \end{subfigure}
\caption{Scenario~2. Assessment of finite sample properties for each parameter (Heterogeneous Covariance Structure).}
  \label{fig:finite_sample2}
\end{figure}

\paragraph{Scenario 2 (Heterogeneous covariance structure).}~\\
\begin{figure}[htb!]
  \centering
  %--- Top: Beta ---
  \begin{subfigure}[b]{0.95\textwidth}
    \centering
    \includegraphics[width=\textwidth]{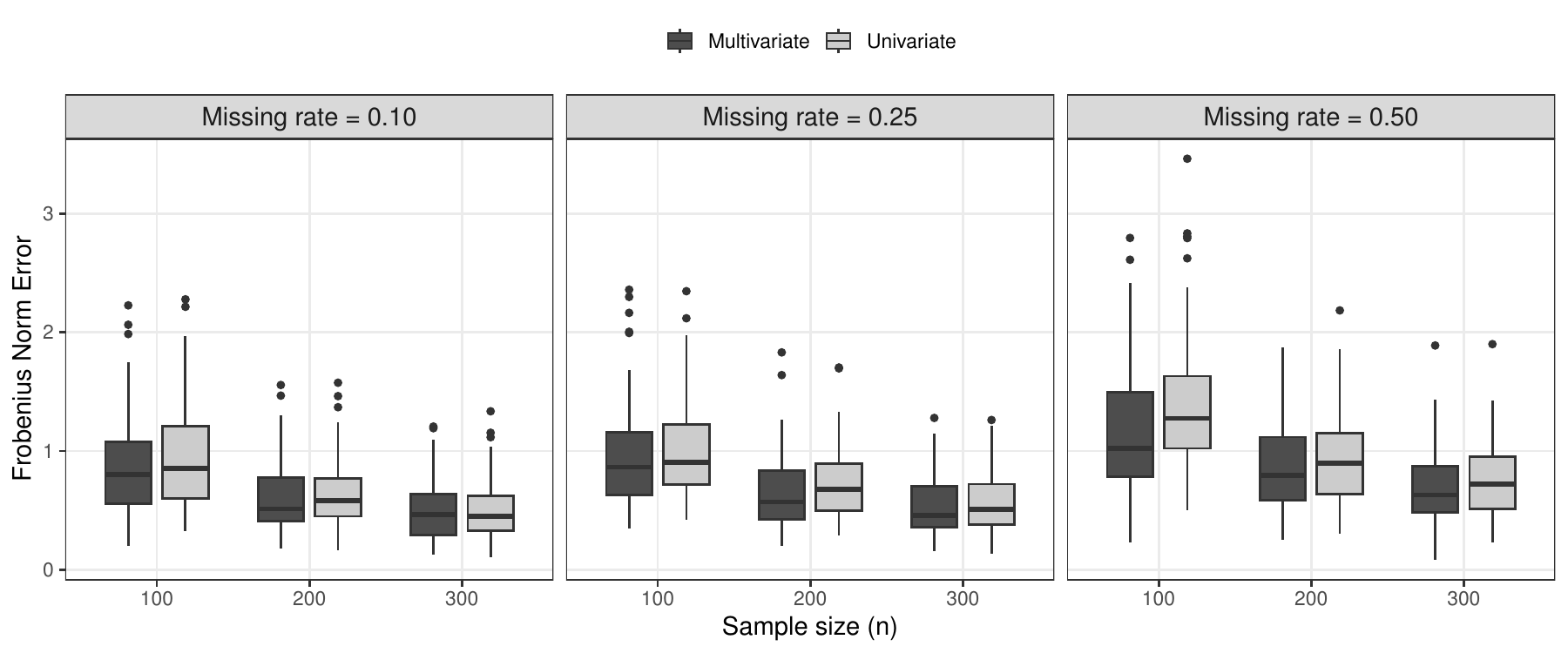}
    \caption{Frobenius norm errors for $\boldsymbol{B}$.}
  \end{subfigure} 
  \vspace{0.5em}
  %--- Bottom: Gamma ---
  \begin{subfigure}[b]{0.95\textwidth}
    \centering
    \includegraphics[width=\textwidth]{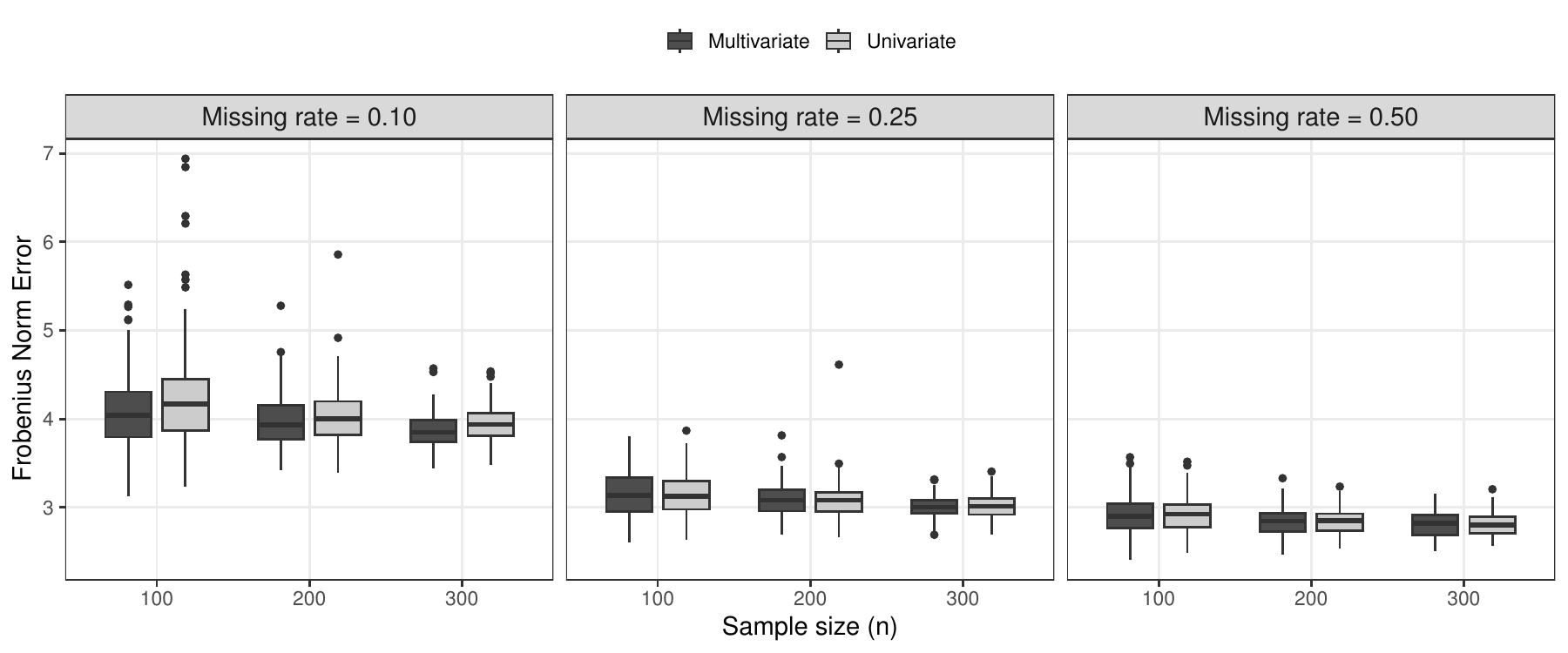}
    \caption{Frobenius norm errors for $\boldsymbol{\Gamma}$.}
  \end{subfigure}
  \caption{Scenario~1. Comparison between the univariate and multivariate models across three levels of sample size and missing rate.}
  \label{fig:uni_comparison_boxplot1}
\end{figure}

\begin{figure}[htb!]
  \centering
  %--- Top: Beta ---
  \begin{subfigure}[b]{0.95\textwidth}
    \centering
    \includegraphics[width=\textwidth]{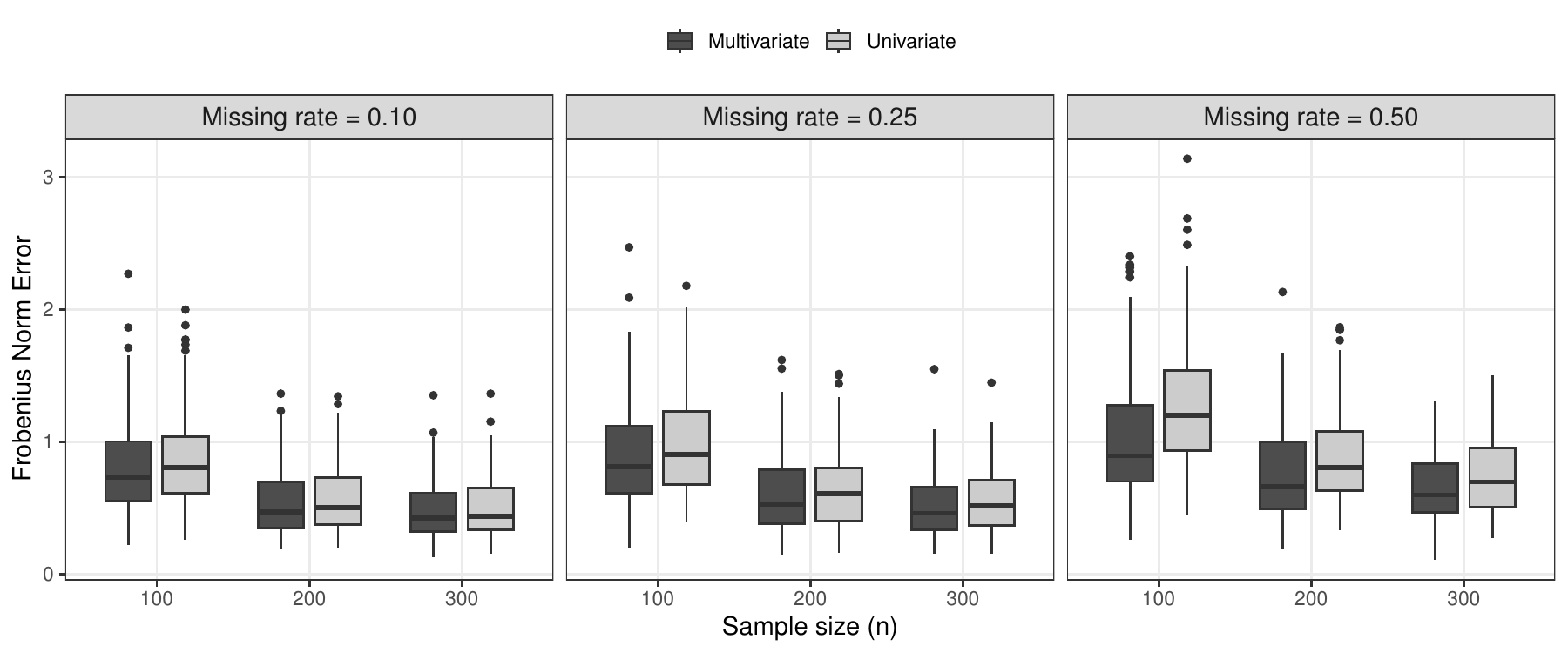}
    \caption{Frobenius norm errors for $\boldsymbol{B}$.}
  \end{subfigure} 
  \vspace{0.5em}
  %--- Bottom: Gamma ---
  \begin{subfigure}[b]{0.95\textwidth}
    \centering
    \includegraphics[width=\textwidth]{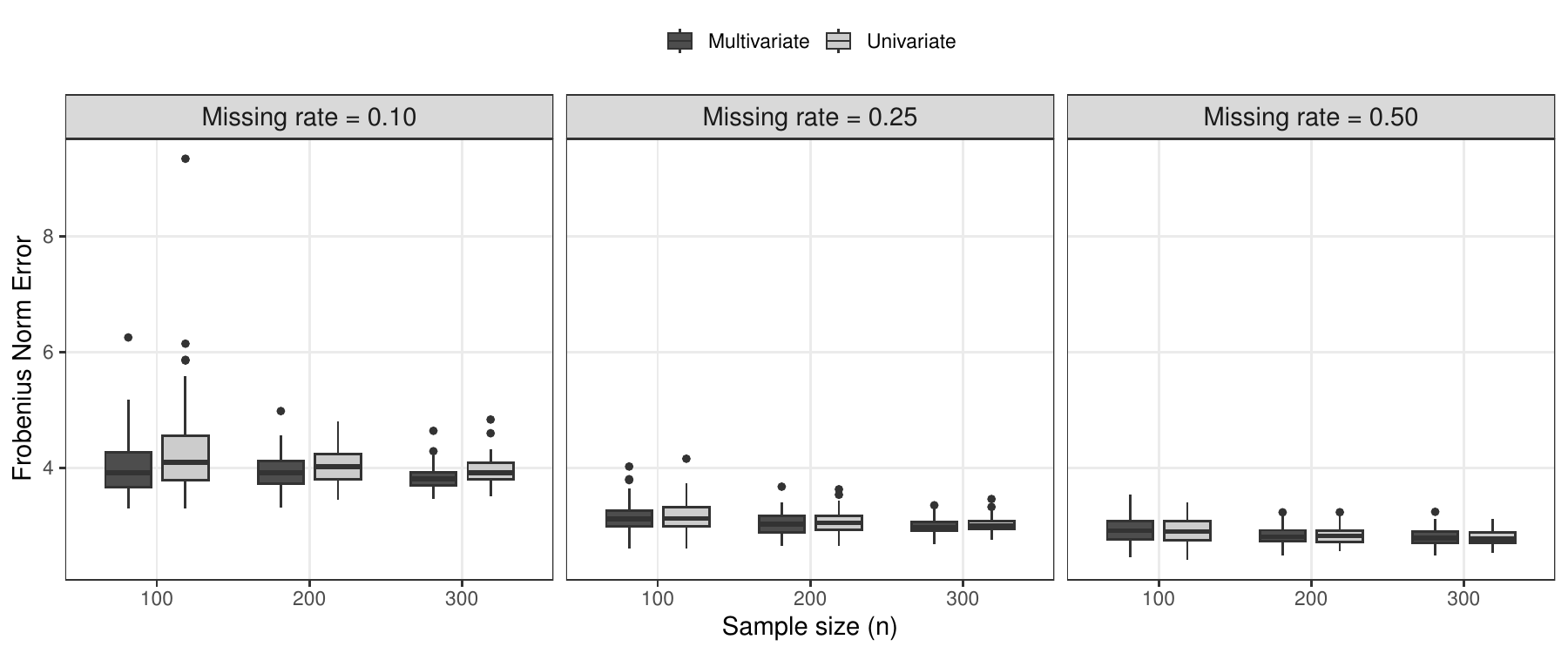}
    \caption{Frobenius norm errors for $\boldsymbol{\Gamma}$.}
  \end{subfigure}
  \caption{Scenario~2. Comparison between the univariate and multivariate models across three levels of sample size and missing rate.}
  \label{fig:uni_comparison_boxplot2}
\end{figure}

To further assess robustness, we consider an alternative setting where the outcome errors exhibit heterogeneous correlation. The outcome designs remain the same, but the correlation matrix $\boldsymbol{\Psi}$ is now specified as:
\[
\boldsymbol{\Psi} =
\begin{pmatrix}
1 & 0.7 & 0.4 \\
0.7 & 1 & 0.1 \\
0.4 & 0.1 & 1
\end{pmatrix}.
\]
This induces unequal correlation strengths across outcome pairs. All other components of the data-generating process, including covariate distributions, $\boldsymbol{B}$, $\boldsymbol{\Gamma}$, and $\boldsymbol{\Sigma}$, remain unchanged.

Figure~\ref{fig:finite_sample2} again shows the estimation results for Scenario~2. The general patterns remain consistent, larger sample sizes improve estimation accuracy, and higher missing rates degrade it. 
%However, compared to Scenario~1, estimation errors for $\boldsymbol{\Gamma}$ and $\phi$ tend to be slightly larger and more dispersed under the heterogeneous correlation structure. This reflects the increased complexity in modeling cross-outcome dependence when correlations are no longer uniform.

Overall, across both scenarios, the ECM estimator demonstrates robustness and efficiency, adapting well to both homogeneous and heterogeneous covariance structures. These findings support the practical utility of the proposed method in recovering parameters of the multivariate Heckman selection model under a wide range of conditions.

\subsection{Comparison with the univariate Heckman selection model}

To benchmark the proposed approach, we also implement the univariate Heckman selection model separately for each outcome variable under the same simulation settings. This comparison allows us to highlight the efficiency gains achieved by jointly modeling the multivariate outcomes. We compare our estimator with univariate fits obtained from the \texttt{HeckmanEM} to assess the benefit of incorporating outcome correlations through the covariance matrix. For each scenario, we simulate 100 replicates and apply the proposed estimation algorithm. The performance is evaluated using the Frobenius norm error between the estimated and true covariance matrices, defined as $\widetilde{\boldsymbol{B}}$, $\widetilde{\boldsymbol{\Gamma}}$.
%$\text{Frobenius}(\widehat{\boldsymbol{\Sigma}}, \boldsymbol{\Sigma}) = \|\widehat{\boldsymbol{\Sigma}} - \boldsymbol{\Sigma}\|_F, \text{Frobenius}(\widehat{\boldsymbol{\Psi}}, \boldsymbol{\Psi}) = \|\widehat{\boldsymbol{\Psi}} - \boldsymbol{\Psi}\|_F$.

Figure \ref{fig:uni_comparison_boxplot1} compares the Frobenius norm errors of the regression parameters $\boldsymbol{B}$ and $\boldsymbol{\Gamma}$ between the univariate and multivariate Heckman selection models across different sample sizes and missing rates. For the outcome parameters $\boldsymbol{B}$, the proposed multivariate model (dark boxes) consistently yields smaller estimation errors than the univariate counterpart (light boxes). The errors decrease as the sample size increases, reflecting improved estimation precision with more information, while higher missing rates lead to modest increases in error. These results indicate that jointly modeling multiple outcomes substantially enhances estimation efficiency for the outcome equations. 

For the selection parameters $\boldsymbol{\Gamma}$, both models show increasing errors as the missing rate rises. This pattern is expected because the selection equation is directly affected by the degree of missingness, making it more difficult to accurately recover the selection mechanism when fewer observations are available. A similar trend has also been reported for the univariate Heckman model in previous studies \citep[e.g.,][]{lim2025heckman, lachos2021heckman}. Nevertheless, the proposed multivariate model maintains lower errors overall, demonstrating robustness in estimating the selection process under incomplete data.

The simulation results confirm that the proposed multivariate Heckman selection model performs robustly under realistic data-generating conditions. 
The model consistently achieves lower estimation errors than the univariate counterpart. These findings demonstrate its advantage over naive approaches that ignore selection bias, highlighting the benefit of jointly modeling multiple correlated outcomes.

%Notably, under Scenario~2, the estimation errors for the selection parameters $\boldsymbol{\Gamma}$ and the correlation structure $\boldsymbol{\Psi}$ are noticeably larger and more dispersed than in Scenario~1, especially at smaller sample sizes and higher missing rates. The Frobenius norm errors for $\boldsymbol{\Gamma}$ show increased variance, and the MSE for the $\phi$ parameter demonstrates a marked rise under heterogeneous correlations. In contrast, the estimation of $\sigma$ and $\rho$ remains relatively stable across scenarios. These differences highlight the increased challenge of accurately recovering heterogeneous correlation structures, yet the ECM estimator still maintains satisfactory performance across all conditions.

\section{Application}

In this section, we present two empirical applications to evaluate the performance of the multivariate Heckman selection model using real-world data. Parameter estimation is conducted using the \texttt{mvHeckman} package developed in \textsf{R}. 

To assess estimation accuracy, we compare the estimated mean matrices obtained from the multivariate model with those from the corresponding univariate models. Furthermore, we implement a bootstrap procedure with 200 replications to estimate the empirical standard errors of the model parameters. 

This approach enables us to quantify the variability of the estimators and to examine the robustness and stability of the proposed method in practical settings.

\subsection{Mroz: Labor supply data}
In this study, we extend the classical econometric dataset of \citet{mroz1984sensitivity} to a multivariate framework. This dataset comprises observations on 753 married white women and their husbands, including detailed information on labor market participation, demographic characteristics, and socioeconomic variables. A key feature of this dataset is the presence of non-random missingness arising from labor force non-participation. Both wages and working hours are observed only for women who participate in the labor market, leading to a sample selection problem. In this sample, the logarithm of wage is missing for 325 individuals (43.2\%) and observed for 428 individuals (56.8\%), with working hours observed for the same subset of individuals.

It has been previously analyzed using both the univariate Heckman selection model \citep{lim2025heckman} and the multivariate contaminated normal censored regression model \citep{lin2025finite}. However, to the best of our knowledge, the multivariate Heckman selection model has not yet been applied to this dataset. To account for the joint determination of wages and working hours, we adopt a bivariate modeling framework that captures dependence between the outcomes while correcting for selection bias induced by labor force participation. We aim to examine how such dependence affects the mean and variability of the estimators, in comparison to the univariate approach.

%Using standard notation, our analysis focuses on the joint modeling of spousal wages based on the classical dataset of \citet{mroz1984sensitivity}. We consider a bivariate response structure to capture the dependence between the wage and hours, accounting for labor force participation and potential selection bias.

%Figure~\ref{fig:ggpairs} presents a pairwise plot of the variables \texttt{wage} and \texttt{hours}, stratified by labor force participation status. The plot includes scatterplots, histograms, and Pearson correlation coefficients. Notably, the overall correlation between wage and hours appears moderately positive (Corr = 0.423), suggesting that higher wages are associated with longer working hours. However, when restricted to individuals who actually participated in the labor force (\texttt{participation} = \texttt{yes}), the correlation becomes weakly negative (Corr = -0.098). This reversal highlights the presence of selection bias: the relationship between observed wages and hours is conditional on participation, and failing to account for this can lead to misleading conclusions about labor supply behavior.

For this empirical application, we consider $R = 2$ outcomes. 
Specifically, let $y_{1i} = \log(\texttt{wage}_i)$ denote the wife's hourly wage (in 1975 U.S. dollars), 
and $y_{2i} = \log(\texttt{hours}_i)$ denote the wife's total annual hours worked.
For $r = 1,2$, the covariate vectors for the outcome equations are defined as
$\boldsymbol{x}_{ri} = (1, \texttt{educ}_i, \texttt{city}_i)$.
The corresponding selection covariates are given by
$\boldsymbol{w}_{ri} = (\boldsymbol{x}_{ri}, \texttt{hwage}_i, \texttt{youngkids}_i, \texttt{tax}_i, \texttt{feduc}_i).$

To obtain standard errors for the estimated parameters, we implemented a nonparametric bootstrap procedure. Bootstrap standard errors were computed as the empirical standard deviations of the bootstrap estimates, and $95\%$ confidence intervals were constructed using the percentile method based on the $2.5\%$ and $97.5\%$ quantiles of the bootstrap distributions. For each replication, we computed the parameter estimates, and the standard errors were calculated as the standard deviations of the 200 bootstrapped estimates. The estimated parameters of the mean matrix $\mathbf{M}$ for both models, UH and MH, are presented below.

{\footnotesize
\[
\mathbf{M}_{\text{UH}} =
\begin{array}{c@{\qquad}c}
    &
    \begin{array}{cc@{\qquad}cc}
        \texttt{log(wage)} &  & \texttt{log(hour)} &  \\
        \texttt{ME} & \texttt{SE} & \texttt{ME} & \texttt{SE}
    \end{array}
    \\[8pt]

    \begin{array}{c}
        \textbf{Outcome model} \\
        \texttt{Intercept} \\
        \texttt{educ} \\
        \texttt{city} \\
        \textbf{Selection model} \\
        \texttt{Intercept} \\
        \texttt{hwage} \\
        \texttt{youngkids} \\
        \texttt{tax} \\
        \texttt{feduc} \\
        \texttt{educ} \\
        \texttt{city}
    \end{array}
    &
    \left[
    \begin{array}{cc}
    \\
        0.669 & 0.239 \\
        0.066 & 0.018 \\
        0.107 & 0.082 \\
         &  \\
        3.802 & 0.764 \\
        -0.104 & 0.015 \\
        -0.415 & 0.078 \\
        -5.782 & 0.847 \\
        -0.020 & 0.013 \\
        0.112 & 0.024 \\
        -0.040 & 0.107
    \end{array}
    \right]
    \!\!
    \left[
    \begin{array}{cc}
    \\
        9.071 & 0.299 \\
        -0.120 & 0.023 \\
        0.062 & 0.110 \\
         &  \\
    1.950 & 0.520 \\
        -0.074 & 0.010 \\
        -0.209 & 0.042 \\
        -3.622 & 0.521 \\
        -0.015 & 0.009 \\
        0.113 & 0.024 \\
        -0.034 & 0.096
    \end{array}
    \right]
\end{array}
\]
}
{\footnotesize
\[
\mathbf{M}_{\text{MH}} =
\begin{array}{c@{\qquad}c}
    &
    \begin{array}{ccc|ccc}
        \multicolumn{3}{c|}{\texttt{log(wage)}} & \multicolumn{3}{c}{\texttt{log(hour)}} \\
        \makebox[2.5em][c]{\texttt{ME}} & 
        \makebox[2.5em][c]{\texttt{SE}} & 
        \makebox[6em][c]{\texttt{CI}} & 
        \makebox[2.5em][c]{\texttt{ME}} & 
        \makebox[2.5em][c]{\texttt{SE}} & 
        \makebox[6em][c]{\texttt{CI}}
    \end{array}
    \\[10pt]
    \begin{array}{c}
        \textbf{Outcome model} \\
        \texttt{Intercept} \\
        \texttt{educ} \\
        \texttt{city} \\
        \textbf{Selection model} \\
        \texttt{Intercept} \\
        \texttt{hwage} \\
        \texttt{youngkids} \\
        \texttt{tax} \\
        \texttt{feduc} \\
        \texttt{educ} \\
        \texttt{city}
    \end{array}
    &
    \left[
    \begin{array}{ccc|ccc}
        -0.461 & 0.220 &(-0.86, -0.01) & 6.812 & 0.299 &(6.55, 7.7) \\
        0.123 & 0.014 &(0.09, 0.15) & -0.015 & 0.022  &  (-0.07, 0.01)\\
        0.047 & 0.056 &(-0.05, 0.15) & -0.037 & 0.097 & (-0.19, 0.17) \\
        &&&\\
        2.722& 1.189 &(2.03, 6.7)& 3.896& 1.175 & (0.78, 5.52)\\
        -0.101 & 0.024 &(-0.18, -0.09)&-0.157 & 0.025&  (-0.16, -0.07)  \\
        -0.434 & 0.113 &(-0.81, -0.34)&-0.785 & 0.123 &  (-0.79, -0.31)  \\
        -4.419 & 1.307 &(-9.23, -4.03)&-6.799 & 1.324 & (-7.86, -2.64) \\
        -0.007 & 0.017 &(-0.04, 0.02)&-0.005 & 0.018 & (-0.04, 0.03)  \\
        0.106 & 0.031 &(0.07, 0.18)&0.203 & 0.030 &  (0.07, 0.19)  \\
        -0.004 & 0.110 &(-0.24, 0.18)& -0.059 & 0.109 & (-0.28, 0.17) \\
    \end{array}
    \right]
\end{array}
\]
}

Both the UH and MH models indicate that education has a significantly positive effect on wages. 
However, the estimated effect of education on hours worked becomes much smaller in the multivariate model compared to the univariate case. Specifically, the coefficient decreases from $-0.120$ (SE = 0.023) in the univariate model to $-0.028$ (SE = 0.022) in the multivariate model. A similar pattern is observed for the city variable, which indicates whether the respondent resides in an urban area (1 = city, 0 = non-city). Its estimated effect on hours worked changes from $0.062$ (SE = 0.110) in the univariate model to $-0.024$ (SE = 0.097) in the multivariate model. In both cases, the estimates are not statistically significant, suggesting that urban residence does not have a meaningful association with annual hours worked when controlling for other factors. Furthermore, the selection equations reveal consistent patterns: women with young children, higher spousal income, and higher tax burdens are less likely to participate in the labor force, while higher levels of own education are associated with greater participation.

To examine the dependence structure between outcomes and selection, we compare two correlation matrices obtained from the multivariate Heckman model: $\Sigma$ and $\Psi$. According to $\Sigma$, the covariance between $\log(\texttt{wage})$ and $\log(\texttt{hours})$ is $0.180$, with a correlation coefficient of $\rho = 0.185$ and a standard deviation of $\sigma = 0.975$, indicating a modest positive association. In contrast, the selection correlation $\Psi_{12} = 0.605$ reveals a much stronger dependence between the outcomes driving labor force participation. These results underscore the importance of modeling joint selection to avoid biased inference.

{\footnotesize
\[
\begin{array}{c}
\begin{array}{@{}c@{}c@{}c@{}|@{}c@{}c@{}c@{}}
\makebox[3em][r]{\texttt{ME}} & 
\makebox[3em][r]{\texttt{SE}} & 
\makebox[6em][c]{\texttt{CI}} & 
\makebox[3em][r]{\texttt{ME}} & 
\makebox[3em][r]{\texttt{SE}} & 
\makebox[6em][c]{\texttt{CI}}
\end{array}
\\[4pt]
\Sigma =
\begin{pmatrix}
0.667 & 0.11 & (0.77,1.16) & 0.239 & 0.10 & (-0.02,0.34) \\
0.239 & 0.10 & (-0.02,0.34) & 1 & & 
\end{pmatrix}
\\[12pt]
\begin{array}{@{}c@{}c@{}c@{}|@{}c@{}c@{}c@{}}
\makebox[3em][r]{\texttt{ME}} & 
\makebox[3em][r]{\texttt{SE}} & 
\makebox[6em][c]{\texttt{CI}} & 
\makebox[3em][r]{\texttt{ME}} & 
\makebox[3em][r]{\texttt{SE}} & 
\makebox[6em][c]{\texttt{CI}}
\end{array}
\\[4pt]
\Psi =
\begin{pmatrix}
0.655 & 0.04 & (0.59,0.72) & 0.718 & 0.06 & (0.60,0.82) \\
0.718 & 0.06 & (0.60,0.82) & 2.208 & 0.19 & (1.85, 2.57)
\end{pmatrix}
\end{array}
\]
}

%The MH model allows for simultaneous estimation while accounting for residual correlation between spouses’ wages, potentially improving inference over the UH specification.

%This specification allows the model to capture heterogeneous effects across the two outcomes while accounting for residual correlation between them through the covariance matrix \( \boldsymbol{\Sigma} \).

\begin{figure}[htbp]
    \centering
    \includegraphics[width=0.95\textwidth]{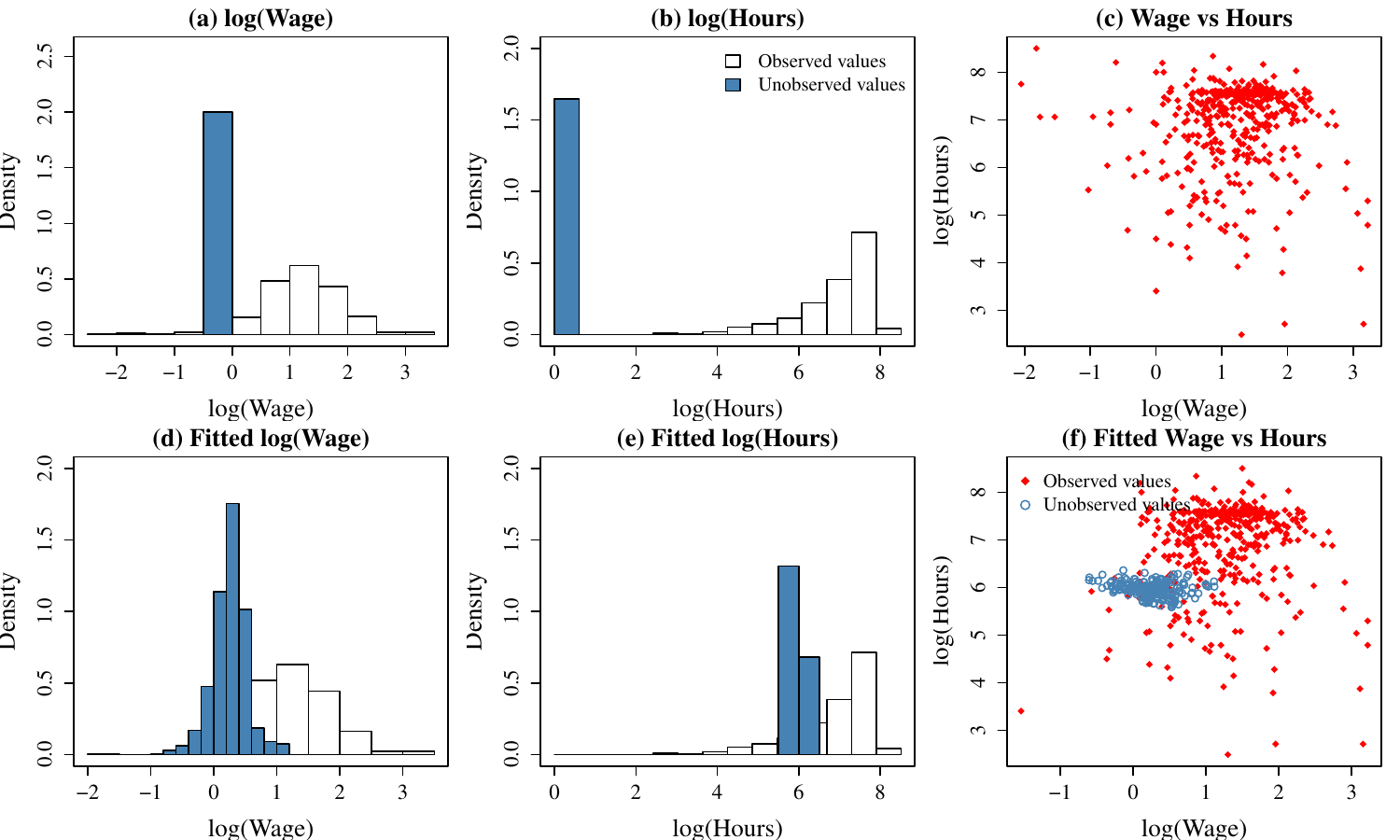}
\caption{
Distribution of observed and censored outcomes and model-implied fitted values. 
Panels (a)–(b) show the distributions of $\log(\text{Wage})$ and $\log(\text{Hours})$ 
for observed and censored observations, while panels (d)–(e) display the fitted distributions from the multivariate Heckman selection model. Panels (c)–(f) present the relationships between wage and hours.
}
\label{fig:ggpairs}
\end{figure}

%Figure~\ref{fig:ggpairs} presents pairwise plots of the variables \texttt{wage} and \texttt{hours}. Since these outcomes are only observed for women who participate in the labor force, the displayed relationship is conditional on labor force participation. The observed correlation between wage and hours is moderately positive (Corr = 0.423), indicating that individuals with higher wages tend to work longer hours in the observed sample.

%However, this relationship should not be interpreted as reflecting the unconditional association in the population. Because wages and hours are observed only for participants, the joint distribution of these variables is subject to sample selection. In particular, individuals with low potential wages or strong household constraints may choose not to participate in the labor market, implying that the observed correlation is shaped by the participation decision. This highlights the importance of explicitly modeling the selection mechanism when analyzing the joint distribution of wages and hours.

Figure~\ref{fig:ggpairs} illustrates the distributions and relationships of the outcome variables under sample selection. Panels (a) and (b) display the empirical distributions of $\log(\text{Wage})$ and $\log(\text{Hours})$ for observed and censored observations. For non-participants, the outcomes are censored and recorded as zeros in the labor market and therefore appear as zeros in the observed data. Panels (d) and (e) present the distributions of the fitted and imputed outcomes obtained from the multivariate Heckman selection model. Compared with the observed distributions, the model-implied outcomes recover a broader range of the latent outcome distribution. Panels (c) and (f) show the relationships between wage and hours for the observed and model-implied outcomes. 
The observed sample exhibits only a weak marginal correlation between the two variables. However, this correlation is calculated conditional on labor market participation and therefore does not reflect the underlying dependence structure in the population. The multivariate selection model accounts for this sample selection and captures the dependence between the outcome equations through the latent error correlation.

Throughout this analysis, we aim to assess whether the proposed multivariate Heckman selection model provides a more accurate and robust characterization of the joint wage process compared to the traditional univariate specification. This specification allows the model to capture heterogeneous effects across the two outcomes while accounting for correlation between them. 

\subsection{NHANES: Blood pressure and health behavior}

To demonstrate the generalizability of our model to biomedical applications, we use publicly available data from the National Health and Nutrition Examination Survey (NHANES) conducted by the Centers for Disease Control and Prevention (CDC) \citep{nhanes}. The dataset includes physical examination and lifestyle information for a nationally representative sample of U.S. adults.

Our primary outcomes are systolic and diastolic blood pressure (SBP and DBP, respectively), measured during clinical visits. Missingness in these variables is substantial due to non-participation or measurement error, making them suitable for selection modeling. For this application, we define the bivariate response as $y_{1i} = \texttt{SBP}_i$ and $y_{2i} = \texttt{DBP}_i$. For $r = 1,2$, the covariates in the outcome equation include age (years), gender (male = 1), race/ethnicity (non-White = 1), body mass index (BMI), and the family poverty-income ratio (PIR). The selection equation includes age, gender, race/ethnicity, and PIR as covariates.

The estimated parameters of the mean matrix $\mathbf{M}$ for both models are summarized below. For SBP, age exhibits the strongest positive association, consistent with well-established clinical evidence that blood pressure increases with aging. Gender and race are also significant predictors, with males and non-White individuals showing higher SBP on average. BMI is positively associated with both SBP and DBP, reflecting its role as a key metabolic risk factor. PIR shows relatively modest effects.

Comparing UH and MH estimates reveals notable differences in both magnitude and uncertainty. In particular, the multivariate model generally yields more stable estimates with reduced variability, reflecting efficiency gains from jointly modeling correlated outcomes. Similar patterns are observed for DBP, although the magnitude of covariate effects differs, indicating outcome-specific relationships.

%BMI was included only in the selection model as an exclusion restriction. This variable affects the likelihood of having valid blood pressure measurements due to behavioral or physical factors but is assumed to have no direct effect on the mean structure of SBP and DBP after adjusting for age, gender, race, and PIR.

%For the selection equation, we include age, gender, race/ethnicity, and PIR, but exclude BMI. The missingness of outcomes in this dataset is largely attributable to behavioral, demographic, and environmental factors rather than to metabolic status such as BMI. Therefore, BMI is excluded from the selection equation.

%The estimated parameters of the mean matrix $\mathbf{M}$ for both models are summarized below. For SBP, age exhibits the strongest association, with an estimated increase of approximately $..$ mmHg per unit increase in the age covariate, reflecting the typical rise in blood pressure associated with aging. Both gender and race are also significant predictors; males and non-White individuals tend to have higher SBP. PIR shows a small positive association. Also, we obtained the standard errors using the same nonparametric bootstrap procedure described in Section 6.1

{\footnotesize
\[
\mathbf{M}_{\text{UH}} =
\begin{array}{c@{\qquad}c}
    &
    \begin{array}{cc@{\qquad}cc}
        \texttt{SBP} &  & \texttt{DBP} &  \\
        \texttt{ME} & \texttt{SE} & \texttt{ME} & \texttt{SE}
    \end{array}
    \\[8pt]

    \begin{array}{c}
        \textbf{Outcome model} \\
        \texttt{Intercept} \\
        \texttt{age} \\
        \texttt{gender} \\
        \texttt{non-White} \\
        \texttt{PIR} \\
        \texttt{BMI} \\
        \textbf{Selection Model} \\
        \texttt{Intercept} \\
        \texttt{age} \\
        \texttt{gender} \\
        \texttt{non-White} \\
        \texttt{PIR} \\
    \end{array}
    &
    \left[
    \begin{array}{cc}
    \\
    77.24&0.96\\
    0.61&0.01\\
    3.47&0.47\\
    3.23&0.50\\
   -0.45&0.15\\
    0.40&0.02\\
    \\
    0.05&0.05\\
    0.00&0.00\\
    0.03&0.03\\
    0.35&0.04\\
    0.11&0.01
    \end{array}
    \right]
    \!\!
    \left[
    \begin{array}{cc}
    \\
    55.69 & 0.94\\
    0.08 & 0.01\\
    1.65 &0.47\\
    1.98 &0.49\\
    0.28 &0.15\\
    0.39 &0.03\\
    \\
    -0.02 &0.05\\
    0.02&0.00\\
    0.09&0.03\\
    0.03&0.04\\
    0.00&0.01
    %sigma
    \end{array}
    \right]
\end{array}
\]
}

{\footnotesize
\[
\mathbf{M}_{\text{MH}} =
\begin{array}{c@{\qquad}c}
    &
    \begin{array}{cc}
        \begin{array}{cc}
            \texttt{SBP} & \\
            \texttt{ME} & \texttt{SE}
        \end{array}
        &
        \begin{array}{cc}
            \texttt{DBP} & \\
            \texttt{ME} & \texttt{SE}
        \end{array}
    \end{array}
    \\[10pt]
    
    \begin{array}{c}
        \textbf{Outcome model} \\
        \texttt{Intercept} \\
        \texttt{age} \\
        \texttt{gender} \\
        \texttt{non-White} \\
        \texttt{PIR} \\
        \texttt{BMI} \\
        \textbf{Selection Model} \\
        \texttt{Intercept} \\
        \texttt{age} \\
        \texttt{gender} \\
        \texttt{non-White} \\
        \texttt{PIR} \\
    \end{array}
    &
    \left[
    \begin{array}{ccc|ccc}
    &&&\\
82.39& 0.93&(82.21,85.75) & 44.39& 0.78&(43.18,46.07))\\
0.57 & 0.01&(0.54,0.59) & 0.22  &0.01&(0.20,0.24)  \\
3.07 & 0.41 &(2.28,3.72) & 2.38 &0.31&(1.84,2.99)  \\
3.29 & 0.41 &(2.33,3.88) & 2.71  &0.34&(1.98,3.26)  \\
-0.61& 0.13&(-0.90,-0.43) & 0.46  &0.10&(0.26,0.66) \\ 
0.37 &0.03 &(0.30,0.42) &0.34 & 0.02&(0.29,0.38) \\
\\
-0.28 &0.05& (-0.14,0.05)&-0.65& 0.05&(-0.21,-0.02)\\
0.31 & 0.00&(0.02,0.03) &0.22 & 0.00&(0.02,0.02) \\
1.08 & 0.04 &(0.01,0.15) &0.94 & 0.04 &(0.02,0.16)\\
1.40 & 0.04 &(0.03,0.16) &1.02 & 0.04 &(0.02,0.15)\\
0.15 & 0.01 &(-0.01,0.03) &0.15 &0.01 &(-0.01,0.04) \\
    \end{array}
    \right]
\end{array}
\]
}

{\footnotesize
\[
\Sigma = 
\begin{pmatrix}
1.85 &0.05&(0.59,0.72))&0.65&0.10&(0.60,0.82)\\
0.65&0.1&(0.60,0.82)&1&0&(1.85,2.57)\\
\end{pmatrix}
  \quad
\]
\[
\Psi= 
\begin{pmatrix}
159.10 & 0.04 &(0.58,0.72) & 77.33& 0.06 &(0.60,0.82) \\
77.33 &0.06&(0.60,0.82) & 99.17&0.19& (1.85,2.57)
\end{pmatrix}
\]
}
%For DBP, the estimated effects follow a different pattern. Age is positively associated with DBP but with a smaller magnitude (approximately mmHg), consistent with the known physiological behavior that DBP increases until mid-adulthood  and decreases later in life.  BMI shows a positive effect on DBP, in contrast to its effect on SBP, suggesting stronger metabolic contributions. Gender and race exhibit similar directional effects as in SBP but with reduced magnitude.

\begin{figure}[htbp]
    \centering
    \includegraphics[width=0.95\textwidth]{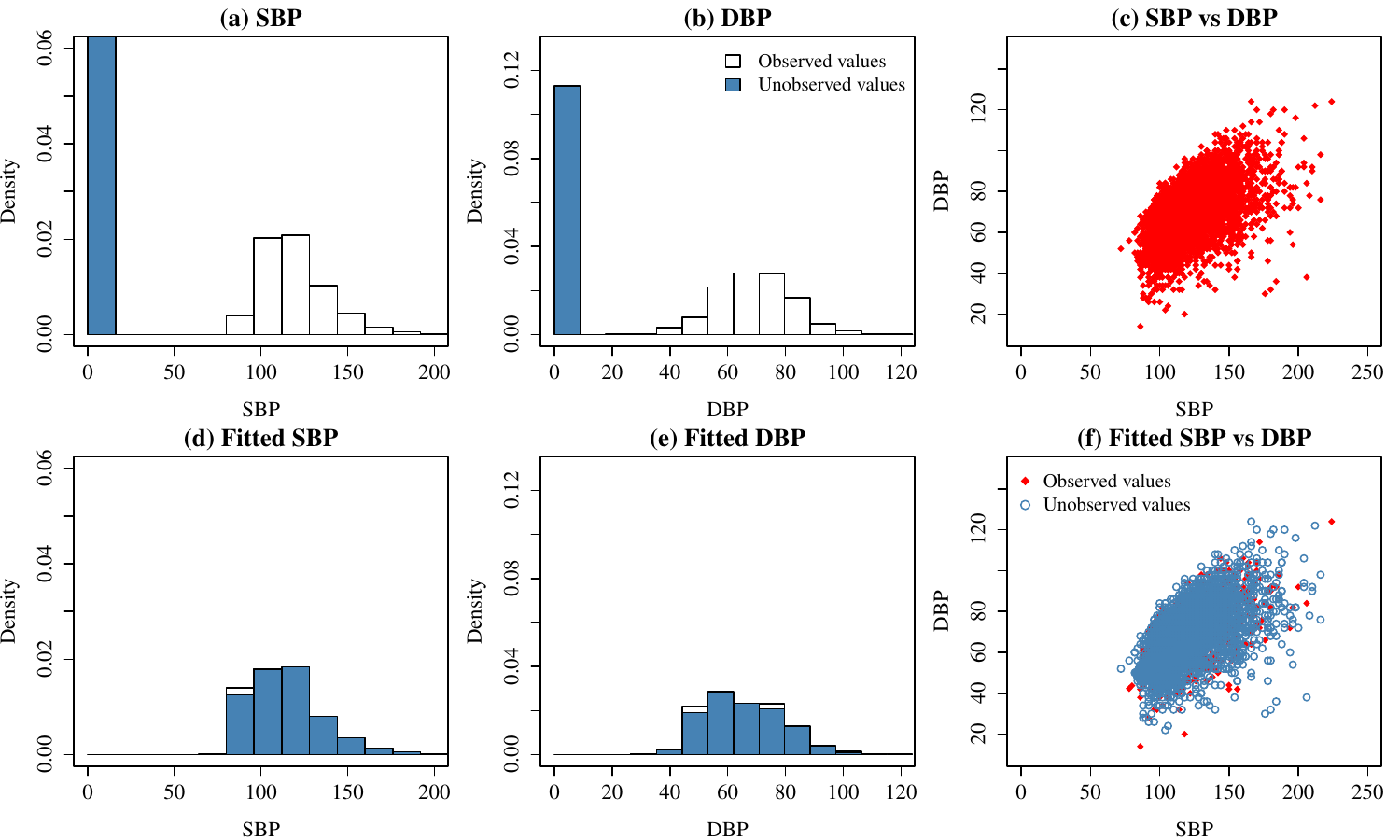}
\caption{
Distribution of observed and censored outcomes and model-implied fitted values. 
Panels (a)–(b) show the distributions of $\log(\text{Wage})$ and $\log(\text{Hours})$ 
for observed and censored observations, while panels (d)–(e) display the fitted distributions from the multivariate Heckman selection model. Panels (c)–(f) present the relationships between wage and hours.
}
\label{fig:application2}
\end{figure}

Figure~\ref{fig:application2} provides additional insight into the distributional features of the data and model fit. Panels (a)–(b) display the distributions of observed and censored outcomes, illustrating the extent of missingness and potential selection bias. Panels (d)–(e) show the fitted distributions under the multivariate Heckman model, demonstrating good agreement with the observed data. Panels (c)–(f) depict the relationship between SBP and DBP, highlighting their positive association and supporting the use of a joint modeling framework.

Overall, the multivariate Heckman model effectively accounts for both the dependence between SBP and DBP and the selection mechanism governing missingness. Compared to separate univariate analyses, the proposed approach provides improved estimation efficiency and a more coherent interpretation of the joint behavior of the outcomes.

%The multivariate framework allows simultaneous estimation of the covariance structure between SBP and DBP while correcting for selection bias due to missing measurements. In this application, the estimated error correlation between SBP and DBP was positive, supporting the clinical understanding that the two outcomes are physiologically linked but not interchangeable. Modeling SBP and DBP jointly yields improved efficiency relative to fitting separate univariate Heckman models and provides a coherent interpretation of their shared and outcome-specific determinants.

%This analysis demonstrates the importance of jointly modeling systolic and diastolic blood pressure in the presence of nonrandom missingness. The proposed multivariate Heckman approach captures both the selection mechanism and the outcome correlation structure, leading to interpretable coefficient estimates and more robust inference for epidemiological applications.

%As in the previous application, we estimate both univariate and multivariate Heckman selection models using a bootstrapped EM procedure, and compare the coefficients and their standard errors across specifications. The goal is to assess whether accounting for cross-outcome dependence improves inference in medical survey data.

\section{Conclusion}

In this paper, we proposed a matrix-variate extension of the classical Heckman selection model to address sample selection bias in the presence of multiple correlated outcomes. By incorporating a separable covariance structure, the proposed framework captures dependence across outcomes while maintaining computational tractability. Parameter estimation was carried out via an ECM algorithm, which demonstrated stable convergence behavior in simulation studies.

Simulation studies under various configurations confirmed that the proposed ECM estimator accurately recovers model parameters, with estimation errors decreasing as sample size increases and increasing with higher missing rates. Although the magnitude of improvement over classical univariate Heckman models was not substantial across all settings, the proposed model consistently achieved lower estimation errors, underscoring the value of jointly modeling multiple outcomes under selection bias.

Empirical analyses based on two real-data applications further illustrated the practical utility of the proposed approach, showing that accounting for cross-outcome dependence can meaningfully influence the estimated effects and their variability. To facilitate practical use, we developed the R package \texttt{mvHeckman}, available on \texttt{GitHub} at \url{https://github.com/heeju-lim/mvHeckman}
, which provides a unified framework for estimation, inference, and simulation of the proposed model.

Looking ahead, future research should extend the current framework to more flexible matrix-variate distributions. In particular, the assumption of matrix-variate normality may limit robustness when the data exhibit departures from symmetry, such as skewness, heavy tails, or multimodality. Promising directions include the development of models based on skewed or heavy-tailed matrix-variate distributions \citep{mahdavi2024robust}, as well as mixture-based formulations that can better capture heterogeneous subpopulations \citep{lachos2017finite, mazza2020mixtures, lin2025finite}. 

%Another important limitation of the proposed framework is that it relies on a selection model formulation. In this approach, the selection mechanism must be explicitly specified, and incorrect specification may lead to biased results. In contrast, pattern-mixture models \citep{little1993pattern} take a different perspective by modeling the outcome distribution separately for different missingness patterns, without requiring a fully specified selection mechanism. Because of this, pattern-mixture approaches can be more robust when the selection process is complex or difficult to model. Therefore, the proposed model may be particularly sensitive to misspecification of the selection equation, especially when the underlying missing data mechanism deviates from the assumed structure. Compared to selection models, pattern-mixture approaches may offer greater robustness to misspecification, whereas selection models provide a more explicit structural interpretation of the selection mechanism. Future research could systematically compare the proposed selection model with alternative frameworks, such as pattern-mixture models, to evaluate their relative performance under different missing data mechanisms and model specifications.

In addition to distributional extensions, Bayesian approaches represent a promising direction for extending the proposed framework beyond the ECM algorithm. While the ECM approach provides computational efficiency through closed-form updates, it may face limitations in fully capturing parameter uncertainty. Recent work by \citet{lim2026bayesian} demonstrates the potential of Bayesian methodologies in the context of univariate Heckman selection models, particularly through the incorporation of prior information, posterior-based inference, and the use of scale mixture representations to accommodate more flexible distributional structures, thereby enabling more comprehensive uncertainty quantification.

Overall, the proposed framework provides a structured and flexible foundation for modeling multiple outcomes under selection bias, and further methodological developments along these directions may substantially broaden its applicability and provide a more robust framework for analyzing complex data environments.

\section*{Appendix}
\subsection*{Proofs of propositions and theorems}

\subsubsection*{Proof of Proposition~\ref{prop:SUN}}
\begin{proof}
The matrix--variate normal specification implies
\[
\boldsymbol{\mathcal{Y}}_i \mid \mathbf{Z}_i
\sim
\mathcal{N}_{2\times R}\bigl(
\mathbf{M}_i=\mathbf{Z}_i\mathbf{B},
\boldsymbol{\Sigma},
\boldsymbol{\Psi}
\bigr),
\]
which is equivalent to
\[
\mathbf{y}_i = \operatorname{vec}(\boldsymbol{\mathcal{Y}}_i^\top) 
\in \mathbb{R}^{2R},
\qquad
\mathbf{y}_i \mid \mathbf{Z}_i \sim
\mathcal{N}_{2R}(\boldsymbol{\mu}_i, \boldsymbol{\Omega}),
\qquad
\boldsymbol{\Omega} = \boldsymbol{\Sigma} \otimes \boldsymbol{\Psi},
\]
where \(\boldsymbol{\mu}_i = \operatorname{vec}(\mathbf{M}_i^\top)\).

Partition
\[
\mathbf{y}_i=
\begin{pmatrix}
\mathbf{y}_{1i}\\[1pt]
\mathbf{y}_{2i}
\end{pmatrix},
\qquad
\boldsymbol{\mu}_i=
\begin{pmatrix}
\boldsymbol{\mu}_{1i}\\[1pt]
\boldsymbol{\mu}_{2i}
\end{pmatrix},
\]
with \(\mathbf{y}_{1i},\boldsymbol{\mu}_{1i}\in\mathbb{R}^R\) and
\(\mathbf{y}_{2i},\boldsymbol{\mu}_{2i}\in\mathbb{R}^R\).
The Kronecker product yields
\[
\boldsymbol{\Omega}_{11}=\sigma^2\boldsymbol{\Psi},
\qquad
\boldsymbol{\Omega}_{22}=\boldsymbol{\Psi},
\qquad
\boldsymbol{\Omega}_{12}
=\boldsymbol{\Omega}_{21}^\top
=\rho\sigma\,\boldsymbol{\Psi}.
\]

\medskip
\noindent
\textit{Restriction to the observed outcomes.}
Let \(R_i^\ast=|\mathcal{S}_i|\) and let
\(\mathbf{S}_i\) be the \(R_i^\ast\times R\) selection matrix that picks
the entries indexed by \(\mathcal{S}_i=\{r: C_{ri}=1\}\).
Define
\[
\mathbf{y}_{1i}^{\text{obs}} = \mathbf{S}_i \mathbf{y}_{1i},
\qquad
\mathbf{y}_{2i}^{\text{obs}} = \mathbf{S}_i \mathbf{y}_{2i},
\qquad
\boldsymbol{\mu}_{1i}^{\text{obs}} = \mathbf{S}_i \boldsymbol{\mu}_{1i},
\qquad
\boldsymbol{\mu}_{2i}^{\text{obs}} = \mathbf{S}_i \boldsymbol{\mu}_{2i},
\]
and
\(
\boldsymbol{\Psi}_i^{\text{obs}} = \mathbf{S}_i \boldsymbol{\Psi}\mathbf{S}_i^\top.
\)
The joint distribution of the observed subvectors is
\[
\begin{pmatrix}
\mathbf{y}_{1i}^{\mathrm{obs}}\\[1pt]
\mathbf{y}_{2i}^{\mathrm{obs}}
\end{pmatrix}
\bigg|\, \mathbf{Z}_i
\sim
\mathcal{N}_{2R_i^\ast}\!\left(
\begin{pmatrix}
\boldsymbol{\mu}_{1i}^{\mathrm{obs}}\\[1pt]
\boldsymbol{\mu}_{2i}^{\mathrm{obs}}
\end{pmatrix},
\begin{pmatrix}
\sigma^2\boldsymbol{\Psi}_i^{\mathrm{obs}} & 
\rho\sigma\,\boldsymbol{\Psi}_i^{\mathrm{obs}}\\[1pt]
\rho\sigma\,\boldsymbol{\Psi}_i^{\mathrm{obs}} & 
\boldsymbol{\Psi}_i^{\mathrm{obs}}
\end{pmatrix}
\right).
\]

The selection pattern $\mathbf{C}_i$ is equivalent to the truncation event
\[
\{\mathbf{C}_i\}
=
\bigl\{
\mathbf{y}_{2i}^{\mathrm{obs}}
\in
\mathcal{A}_i^{\mathrm{obs}}
\bigr\},
\qquad
\mathcal{A}_i^{\mathrm{obs}}
=
\bigl\{\mathbf{y}_2\in\mathbb{R}^{R_i^*} : y_{2,r}>0,\; r\in\mathcal{S}_i\bigr\}.
\]
The observed selection subvector satisfies the Gaussian representation
\begin{equation}
\label{eq:Y2obs_representation}
\mathbf{y}_{2i}^{\mathrm{obs}}
=
\boldsymbol{\mu}_{2i}^{\mathrm{obs}}
+
\boldsymbol{\varepsilon}_{2i}^{\mathrm{obs}},
\qquad
\boldsymbol{\varepsilon}_{2i}^{\mathrm{obs}}
\sim
\mathcal{N}_{R_i^\ast}\!\left(
\mathbf{0},\boldsymbol{\Psi}_i^{\mathrm{obs}}
\right),
\qquad
\boldsymbol{\varepsilon}_{2i}^{\mathrm{obs}} \in \mathcal{A}_i^{\mathrm{obs}}.
\end{equation}

\medskip
\noindent
\textit{Identification of the SUN parameters.}
The joint Gaussian vector
\(
(\mathbf{y}_{1i}^{\mathrm{obs}}, \mathbf{y}_{2i}^{\mathrm{obs}})
\)
together with the truncation event
\(
\mathbf{y}_{2i}^{\mathrm{obs}} \in \mathcal{A}_i^{\mathrm{obs}}
\)
satisfies exactly the block structure required by the unified skew--normal 
(SUN) construction of \citet{arellano2006unification}, where 
$U_0 \in \mathbb{R}^m$ and $U_1 \in \mathbb{R}^d$ in their Eq.~(8) 
correspond here to $\mathbf{y}_{2i}^{\mathrm{obs}} \in \mathbb{R}^{R_i^\ast}$ 
and $\mathbf{y}_{1i}^{\mathrm{obs}} \in \mathbb{R}^{R_i^\ast}$, 
respectively, giving $d = m = R_i^\ast$.
Comparing \eqref{eq:Y2obs_representation} with 
\citet[][Eq.\,(8)]{arellano2006unification}, we identify the SUN 
parameters for the conditional distribution of
\(\mathbf{y}_{1i}^{\mathrm{obs}} \mid (\mathbf{C}_i, \mathbf{Z}_i)\) as

\[
\boldsymbol{\xi}_i = \boldsymbol{\mu}_{1i}^{\mathrm{obs}},
\qquad
\boldsymbol{\Omega}_i = \rho\sigma\,\boldsymbol{\Psi}_i^{\mathrm{obs}},
\qquad
\boldsymbol{\Delta}_i = (\boldsymbol{\Psi}_i^{\mathrm{obs}})^{1/2},
\qquad
\boldsymbol{\tau}_i = \boldsymbol{\mu}_{2i}^{\mathrm{obs}},
\qquad
\boldsymbol{\Gamma}_i = \boldsymbol{\Psi}_i^{\mathrm{obs}},
\]
which establishes that
\[
\mathbf{y}_{1i}^{\mathrm{obs}} \mid (\mathbf{C}_i, \mathbf{Z}_i)
\sim
\mathrm{SUN}_{R_i^\ast, R_i^\ast}(
\boldsymbol{\xi}_i, \boldsymbol{\Omega}_i, \boldsymbol{\Delta}_i,
\boldsymbol{\tau}_i, \boldsymbol{\Gamma}_i),
\]
as stated.
\end{proof}

\begin{proof}[Proof of Corollary~\ref{cor:conditional_expectation_MSLn}]
By Proposition~\ref{prop:SUN}, the conditional distribution of
$\mathbf{y}_{1i}^{\mathrm{obs}} \mid (\mathbf{C}_i, \mathbf{Z}_i)$
belongs to the SUN family. Applying the closed-form expression for 
the first moment of a SUN distribution 
\citep[][Section~2.3]{arellano2006unification}, namely
\[
\mathbb{E}(\boldsymbol{U})
=
\boldsymbol{\xi}
+
\boldsymbol{\Omega}\,\boldsymbol{\Delta}^\top\,
\boldsymbol{\Gamma}^{-1/2}\,
\frac{
\phi_m\!\left(-\boldsymbol{\Gamma}^{-1/2}\boldsymbol{\tau}\right)
}{
\Phi_m\!\left(\boldsymbol{\Gamma}^{-1/2}\boldsymbol{\tau}\right)
},
\qquad
\boldsymbol{U} \sim \mathrm{SUN}_{d,m}(
\boldsymbol{\xi},\boldsymbol{\Omega},\boldsymbol{\Delta},
\boldsymbol{\tau},\boldsymbol{\Gamma}),
\]
and substituting the parameter identification of 
Proposition~\ref{prop:SUN}, namely
$\boldsymbol{\tau}_i = +\boldsymbol{\mu}_{2i}^{\mathrm{obs}}$ and
$\boldsymbol{\Gamma}_i = \boldsymbol{\Psi}_i^{\mathrm{obs}}$, so that
\[
-\boldsymbol{\Gamma}_i^{-1/2}\boldsymbol{\tau}_i = 
-(\boldsymbol{\Psi}_i^{\mathrm{obs}})^{-1/2}\boldsymbol{\mu}_{2i}^{\mathrm{obs}},
\qquad
\boldsymbol{\Gamma}_i^{-1/2}\boldsymbol{\tau}_i = 
+(\boldsymbol{\Psi}_i^{\mathrm{obs}})^{-1/2}\boldsymbol{\mu}_{2i}^{\mathrm{obs}},
\]
we obtain, using $\phi_m(-x) = \phi_m(x)$,
\begin{align*}
\mathbb{E}\!\left(
\mathbf{y}_{1i}^{\mathrm{obs}}
\mid
\mathbf{C}_i,\mathbf{Z}_i
\right)
&=
\boldsymbol{\mu}_{1i}^{\mathrm{obs}}
+
(\rho\sigma\,\boldsymbol{\Psi}_i^{\mathrm{obs}})\,
(\boldsymbol{\Psi}_i^{\mathrm{obs}})^{1/2}\,
(\boldsymbol{\Psi}_i^{\mathrm{obs}})^{-1/2}\,
\frac{
\phi_{R_i^\ast}\!\left(
(\boldsymbol{\Psi}_i^{\mathrm{obs}})^{-1/2}
\boldsymbol{\mu}_{2i}^{\mathrm{obs}}
\right)
}{
\Phi_{R_i^\ast}\!\left(
(\boldsymbol{\Psi}_i^{\mathrm{obs}})^{-1/2}
\boldsymbol{\mu}_{2i}^{\mathrm{obs}}
\right)
}\\[6pt]
&=
\boldsymbol{\mu}_{1i}^{\mathrm{obs}}
+
(\rho\sigma)\,
\underbrace{
(\boldsymbol{\Psi}_i^{\mathrm{obs}})^{1/2}\,
\frac{
\phi_{R_i^\ast}\!\left(
(\boldsymbol{\Psi}_i^{\mathrm{obs}})^{-1/2}
\boldsymbol{\mu}_{2i}^{\mathrm{obs}}
\right)
}{
\Phi_{R_i^\ast}\!\left(
(\boldsymbol{\Psi}_i^{\mathrm{obs}})^{-1/2}
\boldsymbol{\mu}_{2i}^{\mathrm{obs}}
\right)
}
}_{\displaystyle =\, \boldsymbol{\delta}_i^{\mathrm{obs}}(\mathbf{C}_i)}.
\end{align*}
It remains to verify that this expression coincides with $\boldsymbol{\delta}_i^{\mathrm{obs}}(\mathbf{C}_i) = 
\mathbb{E}\!\left(\mathbf{y}_{2i}^{\mathrm{obs}} - \boldsymbol{\mu}_{2i}^{\mathrm{obs}}
\mid \mathbf{C}_i, \mathbf{Z}_i\right)$.
By Proposition~\ref{prop:SUN}, we have
$\mathbf{y}_{2i}^{\mathrm{obs}} \mid \mathbf{Z}_i \sim 
\mathcal{N}_{R_i^\ast}(\boldsymbol{\mu}_{2i}^{\mathrm{obs}}, 
\boldsymbol{\Psi}_i^{\mathrm{obs}})$
truncated to $\mathcal{A}_i^{\mathrm{obs}} = \{\mathbf{y}_2 \in \mathbb{R}^{R_i^*} 
: y_{2,r} > 0,\; r \in \mathcal{S}_i\}$.
Setting $\mathbf{u} = (\boldsymbol{\Psi}_i^{\mathrm{obs}})^{-1/2}
(\mathbf{y}_{2i}^{\mathrm{obs}} - \boldsymbol{\mu}_{2i}^{\mathrm{obs}})$,
we have $\mathbf{u} \mid \mathbf{Z}_i \sim 
\mathcal{N}_{R_i^\ast}(\mathbf{0}, \mathbf{I}_{R_i^\ast})$, and the 
truncation event $\mathbf{C}_i$ becomes
\[
\{\mathbf{C}_i\} 
= 
\bigl\{ \mathbf{u} > \mathbf{a}_i \bigr\},
\qquad
\mathbf{a}_i = -(\boldsymbol{\Psi}_i^{\mathrm{obs}})^{-1/2}
\boldsymbol{\mu}_{2i}^{\mathrm{obs}}.
\]
Since $\mathbf{y}_{2i}^{\mathrm{obs}} - \boldsymbol{\mu}_{2i}^{\mathrm{obs}} 
= (\boldsymbol{\Psi}_i^{\mathrm{obs}})^{1/2}\mathbf{u}$, linearity of 
expectation gives
\[
\mathbb{E}\!\left(
\mathbf{y}_{2i}^{\mathrm{obs}} - \boldsymbol{\mu}_{2i}^{\mathrm{obs}}
\mid \mathbf{C}_i, \mathbf{Z}_i
\right)
=
(\boldsymbol{\Psi}_i^{\mathrm{obs}})^{1/2}\,
\mathbb{E}\!\left( \mathbf{u} \mid \mathbf{C}_i, \mathbf{Z}_i \right).
\]
Since $\mathbf{u}$ has independent standard normal components truncated 
to $u_r > a_{ri}$, the mean of each component follows from the standard 
univariate truncated normal formula
\[
\mathbb{E}(u_r \mid u_r > a_{ri}) 
= 
\frac{\phi(a_{ri})}{1 - \Phi(a_{ri})}.
\]
Substituting $a_{ri} = -(\boldsymbol{\Psi}_i^{\mathrm{obs}})^{-1/2}
\boldsymbol{\mu}_{2i,r}^{\mathrm{obs}}$ and applying
$\phi(-x) = \phi(x)$ to the numerator and $1 - \Phi(-x) = \Phi(x)$ 
to the denominator, we obtain
\[
\mathbb{E}(u_r \mid u_r > a_{ri}) 
=
\frac{
\phi\!\left(
(\boldsymbol{\Psi}_i^{\mathrm{obs}})^{-1/2}\boldsymbol{\mu}_{2i,r}^{\mathrm{obs}}
\right)
}{
\Phi\!\left(
(\boldsymbol{\Psi}_i^{\mathrm{obs}})^{-1/2}\boldsymbol{\mu}_{2i,r}^{\mathrm{obs}}
\right)
}.
\]
Stacking these univariate results into a vector gives
\[
\mathbb{E}\!\left( \mathbf{u} \mid \mathbf{C}_i, \mathbf{Z}_i \right)
=
\frac{
\phi_{R_i^\ast}\!\left(
(\boldsymbol{\Psi}_i^{\mathrm{obs}})^{-1/2}
\boldsymbol{\mu}_{2i}^{\mathrm{obs}}
\right)
}{
\Phi_{R_i^\ast}\!\left(
(\boldsymbol{\Psi}_i^{\mathrm{obs}})^{-1/2}
\boldsymbol{\mu}_{2i}^{\mathrm{obs}}
\right)
},
\]
where $\phi_{R_i^\ast}$ and $\Phi_{R_i^\ast}$ denote the 
$R_i^\ast$-dimensional standard normal density and CDF evaluated 
componentwise, and therefore
\[
\boldsymbol{\delta}_i^{\mathrm{obs}}(\mathbf{C}_i)
=
(\boldsymbol{\Psi}_i^{\mathrm{obs}})^{1/2}\,
\mathbb{E}\!\left( \mathbf{u} \mid \mathbf{C}_i, \mathbf{Z}_i \right)
=
(\boldsymbol{\Psi}_i^{\mathrm{obs}})^{1/2}\,
\frac{
\phi_{R_i^\ast}\!\left(
(\boldsymbol{\Psi}_i^{\mathrm{obs}})^{-1/2}
\boldsymbol{\mu}_{2i}^{\mathrm{obs}}
\right)
}{
\Phi_{R_i^\ast}\!\left(
(\boldsymbol{\Psi}_i^{\mathrm{obs}})^{-1/2}
\boldsymbol{\mu}_{2i}^{\mathrm{obs}}
\right)
},
\]
which coincides with the explicit formula stated in the Corollary, and the poof follows.
\end{proof}

\subsubsection*{Proof of Lemma 1}
{
\begin{proof}
Define $\mathbf{y}=\operatorname{vec}(\boldsymbol{\mathcal{Y}})$, $\bm=\operatorname{vec}(\mathbf{Z}\boldsymbol{B})$ and $\widehat{\mathcal{F}}=\{\widetilde{\mathcal{V}}, \bC, \widehat{\btheta}\}$. By linearity of expectation,
\[
\widehat{\boldsymbol{\Delta}}
= \mathbb{E}\!\left[\Big(\mathbf{y}\mathbf{y}^{\top}-\mathbf{y}\bm^{\top}-\bm \mathbf{y}^{\top}+\bm \bm^{\top}\Big)\,\Big|\,\widehat{\mathcal{F}}\right]
= \mathbb{E}\!\left[(\mathbf{y}-\bm)(\mathbf{y}-\bm)^{\top}\,\Big|\,\widehat{\mathcal{F}}\right].
\]
Hence $\widehat{\boldsymbol{\Delta}}$ is symmetric. For any $\boldsymbol{v}\in\mathbb{R}^{2R}$,
\[
\boldsymbol{v}^{\top}\widehat{\boldsymbol{\Delta}}\,\boldsymbol{v}
= \mathbb{E}\!\left[\big(\boldsymbol{v}^{\top}(\mathbf{y}-\bm)\big)^{2}\,\Big|\,\widehat{\mathcal{F}}\right]>0.
\]
Therefore, $\widehat{\boldsymbol{\Delta}} \succ 0$, since for every nonzero $\boldsymbol{v}$,
$\mathbb{P}(\boldsymbol{v}^{\top}(\mathbf{y}-\bm)= 0\mid \widehat{\mathcal{F}})=0$ a.s.
\end{proof}
}

\subsubsection*{Proof of Theorem 1}

\begin{proof}

First, recall from Lemma \ref{lemma1} that $\widehat{\bDelta}^{\ka *}_i = \widehat{\LL}^{\ka}_{i}\widehat{\LL}^{\ka\top}_{i}$, since $\widehat{\bDelta}^{\ka *}_i$ is  a positive definite matrix. Thus,
\begin{align*}
	{Q(\widehat{\boldsymbol{B}}^{(k+1)}, {\bSigma}, {\bPsi}\mid\widehat{\btheta}^{\ka})} 
    &=  - n \log |\bPsi| - \frac{nR}{2} \log|\bSigma|-\frac{1}{2}\sumas\tr\left[(\bPsi\otimes\bSigma)^{-1} \widehat{\bDelta}^{\ka*}_i\right]\\&=  - n \log (|\bPsi|) - \frac{nR}{2} \log(|\bSigma|)-\frac{1}{2}\sumas\tr\left[(\bPsi\otimes\bSigma)^{-1} \widehat{\LL}^{\ka}_{i}\widehat{\LL}^{\ka\top}_{i}\right]\\
	&= - n \log |\bPsi| - \frac{nR}{2} \log|\bSigma|-\frac{1}{2}\sumas\tr\left[\widehat{\LL}^{\ka\top}_{i}(\bPsi\otimes\bSigma)^{-1} \widehat{\LL}^{\ka}_{i}\right]\\
    &= - n \log |\bPsi| - \frac{nR}{2} \log|\bSigma|-\frac{1}{2}\sumas\sum_{j=1}^{2R}\left[\widehat{\LL}^{\ka\top}_{ij}(\bPsi\otimes\bSigma)^{-1} \widehat{\LL}^{\ka}_{ij}\right],
\end{align*}
where the last equality holds because $\widehat{\LL}^{\ka}_{ij}$ is the $j$th column of the $2R \times 2R$ lower triangular matrix $\widehat{\LL}^{\ka}_{i}$. Applying the trace-vectorization identity, it follows that
\begin{align*}
	{Q(\widehat{\boldsymbol{B}}^{(k+1)}, {\bSigma}, {\bPsi}\mid\widehat{\btheta}^{\ka})} &=
    - n \log |\bPsi| - \frac{nR}{2} \log|\bSigma|  -\frac{1}{2} \sumas \sum_{j=1}^{2R} 
    \big[ \text{vec}(\widehat{\bD}^{\ka}_{ij})^\top (\bPsi\otimes\bSigma)^{-1} \text{vec}(\widehat{\bD}^{\ka}_{ij}) \big] \\
	&=  - n \log |\bPsi| - \frac{nR}{2} \log|\bSigma|-\frac{1}{2}\sumas\sum_{j=1}^{2R}\tr\left[\bSigma^{-1}\widehat{\bD}^{\ka}_{ij}\bPsi^{-1}\widehat{\bD}^{\top\ka}_{ij}\right],
\end{align*}
which concludes the proof.
\end{proof}

%\newpage
%\section*{References}
\bibliographystyle{chicago} 
\bibliography{bibliornl}
%\bibliography{chicago}

\end{document}